\documentclass[sigplan,screen]{acmart}

\usepackage[T1]{fontenc}
\usepackage[utf8]{inputenc}

\usepackage{listings, fancyvrb}
\usepackage{amsmath, amsthm}
\usepackage{graphicx}
\usepackage{subcaption}

\renewcommand\footnotetextcopyrightpermission[1]{}
\newcommand{\myparagraph}[1]{\paragraph{#1}}

\usepackage[framemethod=TikZ]{mdframed}
\mdfdefinestyle{listingstyle}{%
  outerlinewidth=0.25pt,outerlinecolor=black,%
  innerleftmargin=5pt,innerrightmargin=5pt,innertopmargin=0pt,innerbottommargin=0pt%
}

\AtBeginDocument{%
	\providecommand\BibTeX{{%
			\normalfont B\kern-0.5em{\scshape i\kern-0.25em b}\kern-0.8em\TeX}}}

\acmYear{2022}\copyrightyear{2022}
\acmConference[PPoPP '22]{27th ACM SIGPLAN Symposium on Principles and Practice of Parallel Programming}{February 12--16, 2022}{Seoul, Republic of Korea}
\acmBooktitle{27th ACM SIGPLAN Symposium on Principles and Practice of Parallel Programming (PPoPP '22), February 12--16, 2022, Seoul, Republic of Korea}
\acmPrice{15.00}
\acmDOI{10.1145/3503221.3508436}
\acmISBN{978-1-4503-9204-4/22/02}

\renewcommand{\figurename}{Algorithm}

\newcommand{\Guy}[1]{}
\newcommand{\Naama}[1]{}
\newcommand{\Hao}[1]{}
\newcommand{\hedit}[1]{{#1}}

\newcommand{\op}[1]{\texttt{#1}} 

\newcommand{\flock}{\textsc{Flock}}

\newcommand{\locked}{blocking}
\newcommand{\lf}{non-blocking}

\newcommand{\trylock}{tryLock} 
\newcommand{\strictlock}{strictLock} 

\newcommand{\recsucceed}{recursively succeed} 
\newcommand{\recsuccessful}{recursively successful} 

\newcommand{\tlock}{try lock} 
\newcommand{\slock}{strict lock} 

\newcommand{\locklesslock}{lock-free lock}
\newcommand{\blocklock}{blocking lock}

\newcommand{\leaftree}{\op{leaftree}}
\newcommand{\leaftreelf}{\op{leaftree-lf}}
\newcommand{\leaftreelb}{\op{leaftree-bl}}
\newcommand{\arttree}{\op{arttree}}

\newcommand{\leaftreap}{\op{leaftreap}}

\newcommand{\leaftreaplb}{\op{leaftreap-bl}}
\newcommand{\lazylist}{\op{lazylist}}
\newcommand{\hashtable}{\op{hashtable}}

\newcommand{\dlist}{\op{dlist}}
\newcommand{\abtree}{\op{abtree}}

\newcommand{\future}[1]{} 

\newcommand{\bronson}{Bronson}
\newcommand{\natarajan}{Natarajan}
\newcommand{\drachsler}{Drachsler}
\newcommand{\ellen}{Ellen}
\newcommand{\chromatic}{Chromatic}

\makeatletter
\lst@Key{countblanklines}{true}[t]%
{\lstKV@SetIf{#1}\lst@ifcountblanklines}

\lst@AddToHook{OnEmptyLine}{%
	\lst@ifnumberblanklines\else%
	\lst@ifcountblanklines\else%
	\advance\c@lstnumber-\@ne\relax%
	\fi%
	\fi}
\lst@AddToHook{OnEmptyLine}{\vspace{-0.3\baselineskip}}
\makeatother

\newtheorem{theorem}{Theorem}[section]

\newtheorem{definition}{Definition}

\author{Naama Ben-David}
\affiliation{%
	\country{VMware Research, USA}
}
\email{bendavidn@vmware.com}

\author{Guy E. Blelloch}
\affiliation{%
	\country{Carnegie Mellon University, USA}
}
\email{guyb@cs.cmu.edu}

\author{Yuanhao Wei}
\affiliation{%
	\country{Carnegie Mellon University, USA}
}
\email{yuanhao1@cs.cmu.edu}



\begin{document}
\title{Lock-Free Locks Revisited}

	\begin{abstract}

This paper presents a new and practical approach to lock-free locks based on helping, which allows the user to write code using fine-grained locks, but run it in a lock-free manner. 
  Although lock-free locks have been suggested in the past, they are widely viewed as impractical, have some key limitations, and, as far as we know, have never been implemented.
The paper presents some key techniques that make lock-free locks practical and more general. 
The most important technique is an approach to idempotence---i.e. making code that runs multiple times appear as if it ran once.  The idea is based on using a shared log among processes running the same protected code.  
Importantly, the approach can be library based, requiring very little if any change to standard code---code just needs to use the idempotent versions of memory operations (load, store, LL/SC, allocation, free).   

We have implemented a C++ library called \flock{} based on the ideas.  \flock{} allows lock-based data structures to run in either lock-free or blocking (traditional locks) mode.  We implemented a variety of tree and list-based data structures with \flock{} and compare the performance of the lock-free and blocking modes under a variety of workloads.  The lock-free mode is almost as fast as blocking mode under almost all workloads, and significantly faster when threads are oversubscribed (more threads than processors).  We also compare with several existing lock-based and lock-free alternatives.

\end{abstract}

	\maketitle 
	\pagestyle{plain}
	\section{Introduction}

To be or not to be  lock free, that is the question.  Lock-free, or \emph{\lf}, algorithms, 
are guaranteed to make progress even if processes
fault or are delayed indefinitely.   They are, however, burdened with
some issues.
One important issue is that they tend to be significantly more
complicated than their lock-based, \emph{blocking}, counterparts.  Even basic data
structures such as stacks, queues, and singly linked lists can lead to
non-trivial lock-free algorithms with subtle correctness proofs.  More
sophisticated data structures, such as binary trees and doubly linked
lists, become considerably more complicated.  If one needs to
atomically move data among structures, lock-free algorithms become
particularly tricky.  Developing efficient algorithms with fine-grained
locks is not necessarily a cakewalk, but is typically significantly simpler.

Another issue is performance.  The relative performance of \locked{} vs. \lf{}
algorithms depends significantly on the environment in
which they are run.  Several papers demonstrate that \locked{}
concurrent algorithms can be
as fast or faster~\cite{lazylist06,DGT15,Brown18,bwtrees18}.  However, the
experiments described in these papers are typically run in rarified
environments in which all processes are dedicated to the task, often
pinned to dedicated cores.  They are also set up to have no page
faults or other significant delays.  In such environments, it is not
surprising that algorithms using fine-grained (blocking) locks do well.  Some
have noted, however, that in environments with oversubscription (more
processes than cores) \locked{} algorithms can suffer~\cite{DGT15}.
Our experiments verify this (e.g. see the right side of
Figure~\ref{subfig:tree-100M-zipf-200} and~\ref{subfig:tree-100K-zipf}.
Of course, lock-based algorithms can also come to a grinding halt in
environments where processes can be faulty.

In summary, for robustness in mixed environments, or for peace of mind
in general, lock-free algorithms can have a significant advantage, but
they come at the cost of more subtle and complicated designs,
especially when used for more advanced data structures.  Due to the
tradeoffs, there is no universal agreement on whether lock-based or
lock-free algorithms are better---some algorithms are
lock-free~\cite{Natarajan14,BrownER13,harris2001pragmatic,EllenFRB10,bwtrees18,harris2002practical} and others use fine-grained
locks~\cite{drachsler14,bronson10,masstree12,lazylist06,KungL80,Bayer88,arttre16,pugh89}.  A third choice is to use transactional
memory, but this has not yet shown itself to be competitive with
either lock-free or lock-based approaches.

In this paper, we describe and study an approach that can get the best
of both worlds---i.e., allow one to program with fine-grained locks
while getting efficient lock-free behavior.  We base our approach on a
thirty-year-old idea of lock-free locks by Turek, Shasha and
Prakash~\cite{turek1992locking} and independently
Barnes~\cite{barnes1993method} (henceforth the TSP-B approach).   However, we
extend it significantly with several important new ideas to make it
practical and more general.  The high-level idea of the TSP-B approach
is that when a thread takes a lock it leaves behind a descriptor that
allows other threads that want the lock to help it complete its
protected code and free its lock.  The general idea of using
descriptors for helping is now widely used in the implementation of
specific lock-free applications, such as
multiword-CAS~\cite{guerraoui2020efficient,feldman2015wait,harris2002practical, wang2018easy},
other multiword operations~\cite{BrownER13}, software transactional
memory~\cite{shavit1997software,fraser2007concurrent,ramalhete2019onefile}, and specialized
data structures~\cite{EllenFRB10,shafiei2014non,winblad2021lock,disc2021mvgc,dice2002mostly}.

Despite the use of descriptors for helping in specific applications,
we know of no general implementations of lock-free locks.  Most of the
papers cited above mention the TSP-B approach, often as motivation for
their more specific approach, but describe it as impractical.  The
issue is that the TSP-B approach requires translating code in the lock
into a form such that every read or write effectively requires saving
the context of the process (program counter and local variables)
so that others can help it run from that point.  Such code can be very
inefficient even when no helping occurs.  Equally importantly, it makes
the approach very difficult and clumsy to use without a
special-purpose compiler.  Their approach also constrains the code
inside the locks to only allow race-free reads and writes to shared
memory. 

The key contribution of this paper is an approach to avoid the
``context-saving'' on each memory operation, making the approach
practical, and additionally making it more general.  In our approach
the user can write standard code based on fine-grained locks, and
using a library interface, get lock-free behavior.  Beyond being
efficient and offering a simple library-based interface, our approach
generalizes the TSP-B approach by (1) allowing races in the locked
code, (2)
supporting memory allocation and freeing in the locked code, and (3)
supporting try locks, which we demonstrate are much more efficient
than the standard strict locks.
The advantage of try lock is that it returns false if the lock is currently taken, giving the user the flexibility of either trying again or performing a different operation.

Our approach is based on a new technique to achieve idempotence.
Intuitively, idempotent code is code that can be run multiple times but
appears to have run once~\cite{idempotence02,idempotence12,idempotence13,lltheory}.  Such code is important in the
TSP-B approach since multiple helpers could run the same locked code
when helping.  TSP-B suggest particular approaches to
idempotence (the approaches by TSP and B are quite similar) but failed
to abstract out the notion of just needing idempotent code.  Here, we
abstract out the need of idempotence for lock-free locks and suggest a
very different, as well as more efficient and general, approach to
achieving idempotence.  We also point out that to nest locks, we simply
need the locking code itself to be idempotent, leading to
locking code that is very simple.  

In our approach to idempotence, instead of using ``context saving''
we maintain a shared log among processes running the same code.  The
log keeps track of  all reads from shared mutable locations, as well
as some other events, such as memory allocations. Whenever a copy of
the thunk executes a loggable operation, it commits it to the log
using a compare-and-swap (CAS).  Whichever copy commits first wins,
and all others take the value committed instead of their attempted
commit.  In this way, they all see the same committed values, e.g.,
the same reads, even though they are running in an arbitrary
interleaved manner.


One key advantage of our approach is that the user can write
concurrent algorithms based on fine-grained locks, and then either run
them in a \emph{lock-free mode} (with helping) or a \emph{blocking mode} (no helping).
The blocking mode can use a standard lock implementation
without logging.  The helping mode will log, at some additional cost,
but guarantee lock-free behavior.
Another key advantage over TSP-B is that our approach is based on try
locks, instead of strict locks, which turns out to be important for
the efficiency of 
optimistic use of fine-grained locks.

We have implemented our approach as a C++-based library called
\flock{}.  Based on the library we have implemented several data
structures based on try-locks, including singly linked lists, doubly
linked lists, binary trees, balanced blocked binary trees, (a,b)-trees,
hash tables, and adaptive radix trees (ART).  We compare performance
of our versions in lock-free mode and blocking mode to the most
efficient existing data structures we found, both lock-based and
lock-free.  The lock-based data structures generally perform slightly
better under controlled environment with one process per processor,
but perform significantly worse when oversubscribing with multiple
processes per processor.  Comparing running our algorithms in
lock-free vs. blocking mode, the lock-free performance rarely has more
than 10\% overhead, and typically much less.  However with
oversubscription the lock-free mode greatly outperforms the blocking
mode by up to 2.4x. 

Our contributions include the following:
\begin{enumerate}
\item We present a new practical approach to achieving idempotence in general
  code, which relies on logging rather than context saving.  

\item
We present a new approach to lock-free try-locks.   They can be nested.

\item
We develop a general library-based interface to support our ideas.

\item 
We compare several existing approaches with ours, both using locking
and without locking.

\item 
We develop the first lock-free implementation of adaptive radix trees.
\end{enumerate}

	\subsection{Example of Using Lock-Free Locks}
\label{sec:listexample}

\newcommand{\btlock}{\texttt{{try\_lock}}}

\newcommand{\wrapred}[1]{\texttt{\color{red}\textbf{#1}}}
\newcommand{\template}[2]{\wrapred{#1<}#2\wrapred{>}}
\newcommand{\mutable}[1]{\template{mutable\_}{#1}}

\begin{figure}
\begin{lstlisting}
struct link {
	@\mutable{link*}@ next;@\label{line:next}@
	@\mutable{link*}@ prev;@\label{line:previous}@
	@\mutable{bool}@ removed;@\label{line:removed}@
	Key k;  Value v; @\wrapred{lock}@ lck; 
	link(Key k, Value v, link* next, link* prev) 
		: k(k), v(v), next(next), prev(prev), removed(false) 
		{};};

link* find_link(link* head, Key k) {
	link* lnk = (head->next).load();
	while (k > lnk->key) lnk = (lnk->next).load(); @\label{line:loadFind}@
	return lnk;}

std::optional<Value> find(link* head, Key k) {
	link* lnk = find_link(head, k);
	if (lnk->key == k) return lnk->value;   // found
	else return {}; }                         // not found

bool insert(link* head, Key k, Value v) {
	while (true) {
		link* next = find_link(head, k);
		if (next->key == k) return false;  // already there
		link* prev = (next->prev).load(); @\label{line:loadInsert}@
		if (prev->key < k &&
		    @\wrapred{try\_lock}@(prev->lck, [=] {
					if (prev->removed.load() ||   // validate
					    (prev->next).load() != next)
						return false;
					link* newl = @\template{allocate}{link}@(k, v, next, prev);
					prev->next = newl;  @\label{line:insertnext}@// splice in
					next->prev = newl;  @\label{line:insertprev}@
					return true;}))
			return true;}};   // success

bool remove(link* head, Key k) {
	while (true) {
		link* lnk = find_link(head, k);
		if (lnk->key != k) return false;  // not found
		link* prev = (lnk->prev).load(); @\label{line:loadRemove}@
		if (@\wrapred{try\_lock}@(prev->lck, [=] { @\label{line:testacquired}@
					return @\wrapred{try\_lock}@(lnk->lck, [=] {
						if (prev->removed.load() ||    // validate @\label{line:removetest}@
						    (prev->next).load() != lnk) @\label{line:sametest}@
							return false; 
						link* next = (lnk->next).load();
						lnk->removed = true;
						prev->next = next; @\label{line:remnext}@  // splice out
						next->prev = prev; @\label{line:remprev}@
						@\template{retire}{link}@(lnk);
						return true;});})) 
			return true;}}  // success
\end{lstlisting}
\caption{Sorted doubly-linked lists using fine-grained
  optimistic locks with \flock{}.
\flock{} code shown in red.
}
\label{alg:list}
\end{figure}

To be concrete on how lock-free locks are used in our framework, we give an example of maintaining a concurrent sorted doubly-linked list supporting insert, delete, and find.  The example uses optimistic fine-grained locks~\cite{KungL80,KungR81}.  Our C++ code using \flock{} is given in Algorithm~\ref{alg:list}.  
Each link holds a key and value, a previous and next pointer, a lock,
and a flag indicating whether the link has been removed.  The
\texttt{mutable\_} wrapper around \texttt{next}, \texttt{prev}, and
\texttt{removed} (lines~\ref{line:next}--\ref{line:removed}) indicates
that these are shared mutable values\footnote{The underscore in \texttt{mutable\_} is used
  since unforunately \texttt{mutable} is an obscure reserved symbol in C++.}.  They need to be read using a \texttt{load}, with a similar interface to a C++ \texttt{atomic}.  \flock{} will log loads of mutable values when inside a lock.  Since the key and value are immutable, they need not be put in \texttt{mutable\_}.

Locks are attempted with the \btlock{} function.  It takes a lock as
an argument, as well as a \emph{thunk} (a function with no arguments).
In \flock{}, the thunk is simply a C++ lambda expression containing
the code to be run when the lock is acquired.
If the lock is free, \btlock{} acquires the lock, runs the thunk,
releases the lock, and returns the thunk's return value (a boolean).
Otherwise it returns false.  The \btlock{} function forces locks to be
properly nested.  This is important for our lock-free locks since the
thunk captures the code that might need to be helped by another
\btlock{}.  In Section~\ref{sec:locks} we describe a function that
avoids pure nesting and supports, for example, hand-over-hand locking.

The \texttt{find\_link} finds the first link with a key greater than or equal to the requested key.  It requires no locks.  The \texttt{find} just extracts the value from the link if the key matches.

The \texttt{remove} first finds the link \texttt{lnk} potentially containing the key.
If it does not contain the key, then it returns \texttt{false} indicating the key was not in the list.
Otherwise it tries to acquire a lock on the previous link (\texttt{prev}) and \texttt{lnk}.
If either fails because they are already locked, the condition on line~\ref{line:testacquired} will be false and the \texttt{while} loop will repeat.
The conditions on lines~\ref{line:removetest} and~\ref{line:sametest}  validate that 
the previous link has not been deleted, and \texttt{prev->next} still points to \texttt{lnk}.
If either test fails then the \texttt{while} loop is repeated.   If the tests pass,
the code in the lock loads the next pointer from \texttt{lnk}, marks \texttt{lnk} as removed, splices it out of the doubly linked
list, and retires its memory\footnote{\flock{} uses an epoch based memory manager.   The \texttt{retire} puts the pointer asside and frees
  it when it is safe (after all concurrent operations finish).}.
Note that a lock is not required on \texttt{next}.  This is because a deletion of \texttt{next} or an insertion of an element before \texttt{next} would
require a lock on \texttt{lnk} so it cannot happen concurrently.
The \texttt{insert} is similar to \texttt{remove}.

This locking-based code for doubly-linked lists is much simpler than any lock-free versions we know of~\cite{Gre02,ST08,shafiei2014non,disc2021mvgc,AH13}.  The difficulty in generating a lock-free version based on CAS is that lines~\ref{line:remnext}--\ref{line:remprev} need to be applied atomically, as do lines~\ref{line:insertnext}--\ref{line:insertprev}.  Our approach gives us a lock-free algorithm using the simple lock-based algorithm.  As we show in our experiments, the lock-free version is almost as fast as the locking one without oversubscription, but much faster with oversubscription.

	\section{Model}

We consider an asynchronous shared memory accessed by $n$ processes. Processes can access the shared memory via the following atomic primitives: \emph{read, write} and \emph{compare-and-swap (CAS)}, defined in the standard way.
We also assume a \emph{sysAlloc} which returns an unused block of memory and a \emph{sysRetire} which delays freeing the memory block until it is safe.
Each process also has access to private memory. 

An \emph{execution} is a sequence of \emph{steps}, where each step specifies a primitive, its arguments, its return values, and the executing process. The steps taken by a process in an execution implement \emph{operations}. An \emph{event} is the \emph{invocation} or \emph{response} of an operation, which specify its arguments and return values, respectively, as well as its calling process. 
The first step of an operation in an execution is associated with its invocation and its last step is associated with its response. A \emph{history} is a sequence of events, and can be derived from an execution $E$ by including the invocations and responses of operations in the order their associated steps appear in $E$.  
An execution is \emph{valid} if it is consistent with the semantics of the memory operations.

A data structure is a set of operations. Each operation is specified by a \emph{sequential specification}, which defines its expected behavior in an execution in which the executing process's steps are not interleaved with the steps of any other process. An implementation of a data structure specifies code for processes to run for each of its operations. 
An implementation of a data structure $D$ is \emph{lock-free} if, in any infinite execution in which processes follow this implementation, infinitely many operations complete. This is equivalent to requiring a finite number of steps between responses.

We say a memory location suffers from the \emph{ABA problem} in some implementation if it is possible for the value written on that memory location to go back to what it was at some previous point in some execution of this implementation. We say an implementation suffers from the ABA problem if there is some memory location that suffers from the ABA problem in that implementation. An  implementation is \emph{ABA-free} if it does not suffer from the ABA problem.

	\section{Idempotence}

To achieve lock-free critical sections, processes must be able to help each other. In particular, if some process holds a lock and crashes, others must be able to release the lock. Since it is possible that the crashed process has already begun its critical section, the other processes must complete its critical section for it before releasing the lock. 

This leads to the need to have \emph{idempotent} critical sections. Intuitively, a piece of code is idempotent if, when it is executed multiple times, it only appears to take effect once. 
Thus, if we have idempotent critical sections, processes can safely help execute someone else's critical section, without worrying about who else has also executed it.
Some code is naturally idempotent.  For example, a critical section that contains just one CAS instruction, which does not suffer from the ABA problem, is idempotent. After it is executed for the first time, subsequent executions of it would have their CAS fail, thus leaving the memory in the same state. Many hand-designed lock-free data  structures achieve their lock-freedom by allowing helping in such short, naturally idempotent sections.

In general, however, most code is not idempotent by default. For example, code incrementing a counter would yield different resulting counter values if it is executed several times. 
Thus, general lock-free constructions must be able to make general code idempotent. Several approaches in the literature have shown how to do so~\cite{turek1992locking,barnes1993method,ben2019delay,lltheory}. 
In this section, we define idempotence formally and present a new construction that makes any piece of code idempotent.

\subsection{Idempotence Definition}
A \emph{thunk} is a procedure with no arguments~\cite{ingerman1961thunks}. Note that any procedure with given arguments can be made a thunk by wrapping it in code that reads its  arguments from memory.

We follow the definition of idempotence introduced in~\cite{lltheory}.
A \emph{run} of a thunk $T$ is the sequence of steps on shared data taken by \emph{a single process} to execute or help execute $T$.
The runs of a thunk by different processes can be interleaved and each run may take a different branch through the thunk depending on the memory state that it sees.
A run is \emph{finished} if it reached the end of $T$.
\hedit{We say a sequence of steps $S$
	is \emph{consistent} with a run \emph{r} of \emph{T} if, ignoring process ids, $S$ contains the exact same steps as $r$.}
\hedit{We use $E~|~T$ to denote the result of starting from an execution $E$ and removing any step that does not belong to a run of the thunk $T$.}


\begin{definition}[Idempotence~\cite{lltheory}]
	\label{def:idempotence}
	A thunk $T$ is idempotent if in any valid execution $E$ consisting of runs of $T$ interleaved with arbitrary other steps on shared data, there exists a subsequence $E'$ of $E | T$ such that:
	\begin{enumerate}
		\item
		if there is a finished run of $T$ (response on $T$), then the last \hedit{step} of the first such finished run must be the end of $E'$,
		\item
		removing all of $T$'s steps from $E$ other than those in $E'$ leaves a valid history consistent with a single run of $T$.
	\end{enumerate}
\end{definition}

%

\Guy{This intuitive description no longer matches the definition.   And it still uses operation.}
Intuitively, this definition allows a thunk $T$ to be executed by several processes (in several runs of $T$), but other than one copy of each step executed for $T$, the rest \hedit{are not effectual (i.e. have no impact on the rest of the execution).
Furthermore, after one run of $T$ completes, no other runs of $T$ can execute an effectual step}.


We assume that a thunk may have \emph{thunk-local} memory which can only be accessed by processes executing the thunk. 
In our simulation the log is thunk-local.   Such memory is not ``shared data'' as defined in Definition~\ref{def:idempotence}.


\subsection{Our Approach to Idempotence}
We now present a new approach to achieving idempotence in any code that is ABA-free. We note that it is easy to make code ABA-free by attaching a counter to any memory location that suffers from the ABA problem, and updating that counter every time the value is updated (our implementation does this).  Rather than basing our idempotence construction on \emph{context saving}, as were previous general idempotence constructions, we base our approach on using a shared \emph{log}. We present pseudocode for the approach in Algorithm~\ref{alg:idempotent}.

\begin{figure}
\begin{lstlisting}
type Log = shared<entry>[logSize];
type Thunk = function with no arguments returning bool

private process local:
	Log* log;  // the current log for a process @\label{line:log}@
	int position; // the current position in the log @\label{line:count}@

struct descriptor:
	Log* log;
	Thunk thunk;
	mutable<boolean> done;

descriptor* createDescriptor(Thunk f):
	Log* log = allocate<Log>();
	return allocate<descriptor>(log, f, false);

void retireDescriptor(descriptor* T):
	retire<Log>(D->log);
	retire<descriptor>(D);

bool run(descriptor* D):
	Log* old_log = log; // store existing log and position
	int old_pos = position 
	log = D->log;  // install D's log @\label{line:initLog}@
	position = 0; @\label{line:initPosition}@
	bool returnVal = D->thunk(); // run thunk
	log = old_log; // reinstall previous log and position
	position = old_pos 
	return returnVal;

<V, bool> commitValue(V val):@\label{line:commitvalue}@
	if (log == null): return <val, true>; @\label{line:nolog}@
	bool isFirst = log[position].CAS(empty, val); @\label{line:commit}@
	V returnVal = log[position].read(); @\label{line:commit-load}@
	position++;
	return <returnVal, isFirst>;

struct mutable<V>:
	shared<V> val;
	V load():
		V v = val.read(); @\label{line:loadRead}@
		return commitValue(v).first;
	void store(V newV):
		V oldV = load();
		val.CAS(oldV, newV); @\label{line:storeCAS}@
	void CAM(V oldV, V newV):
		V check = load();
		if (check != oldV): return; @\label{line:CAMcheck}@
		val.CAS(oldV, newV); @\label{line:CAMCAS}@
		
V* allocate<V>(args): @\label{line:allocdef}@
	V* newV = sysAllocate<V>(args); //use system allocator @\label{line:allocate}@
	<obj, isFirst> = commitValue(newV);
	if not isFirst: sysFree<V>(newV); @\label{line:free}@
	return obj;
		
void retire<V>(V* obj) @\label{line:retiredef}@
	<_, isFirst> = commitValue(1);
	if isFirst: sysRetire(obj); @\label{line:retire}@
\end{lstlisting}
\caption{Idempotent primitives.   The \texttt{\bf entry}s of the log are assumed
to hold any type that fits in a word (or two if using double width CAS).  The log is of fixed sized, but
could grow by adding blocks as needed (see Section~\ref{sec:imp} for details).}
\label{alg:idempotent}
\end{figure}

We store each thunk in an struct, called the \emph{descriptor}, that includes the thunk itself, 
and a \emph{log}. 
The log keeps track of all values read, allocated or retired in any execution of the thunk. 
The \texttt{shared<T>} indicates a variable of type \texttt{T} that is shared among processes.
The log, however, is \emph{thunk-local};
any process executing this thunk uses the same descriptor struct, so the log is shared by all processes that execute this thunk,\footnote{Note that this differs from distributed logs used in, for example, optimistic transactional memory~\cite{KungR81}, where each process has its own log.  It also differs from logs used to commit successful
transactions.} but no other process can access this log.

We implement five operations for idempotent code: \emph{load}, \emph{store}, \emph{CAM} (a CAS that does not return any value), \emph{allocate}, and \emph{retire} using the memory primitives \emph{read}, \emph{write}, \emph{cas}, \emph{sysAlloc} 
and \emph{sysRetire}.
Any thunk can then be implemented using these operations.  For ease of use, load, store and CAM are implemented in a struct called \emph{mutable} that can wrap any type.  We call it mutable since these are locations shared among processes that can be modified, yet still have to be idempotent. Any variable declared as mutable automatically uses our implemented operations rather than the corresponding primitives, keeping programmer effort to a minimum. We assume that CAMs and stores cannot race on the same location.  Any non-mutable value, or any local variables/locations can be read and written as usual without using a \emph{mutable}.   For our purposes a value is \emph{non-mutable} (constant) if it is written once (e.g. on initialization) and only read after it is written.

The idea of the approach is that each process keeps track of its current log (line~\ref{line:log}) and how many items it has
logged in it so far while running the corresponding thunk (its \emph{position}, on line~\ref{line:count}). Thus, when it starts executing a new thunk, it initializes its position to 0 and its local log to point to this thunk's log (lines~\ref{line:initPosition} and~\ref{line:initLog}). The process saves its previous log and position so that it can go back to the previous thunk when it finishes executing the new one. This is useful for executing nested thunks.   Once a process has installed its new log and initialized its position, it can start running the thunk.  Whenever it executes a new loggable instruction (load, allocate or retire), it uses the shared log of the thunk to record the return value of this instruction and to see whether others have already logged it.

Values are stored in the log using a helper function called \emph{commitValue} (line~\ref{line:commitvalue}). 
This function takes in a value to be logged; intuitively, this is the intended return value of the current instruction.  The process uses its current position to index into its thunk's log. It tries to commit its value by using a CAS on \texttt{log[position]}, with old value \texttt{empty}, and new value equal to the value it would like to log. All log entries are initialized to \texttt{empty} and we assume that no process attempts to write \texttt{empty} into a mutable variable. The process then checks what value is written in \texttt{log[position]}, and returns this value, as well as a boolean indicating whether or not its CAS succeeded (i.e. whether it was the first to execute this instruction on this thunk). When the process does not currently have a log (i.e. is not currently executing a thunk), the \texttt{commitValue} function simply returns the input value and the success flag set to true (line~\ref{line:nolog}). 
With our locks this happens when the instruction is executed outside of all locks.  For example, no logging is
needed for the loads on line~\ref{line:loadFind} in Algorithm~\ref{alg:list} since they are not in a lock, but
the load on line~\ref{line:removetest} is logged in the descriptor for its surrounding lock.

To load a value from a given mutable variable, a process simply does a read from the location, and then tries to commit its value to the log by calling \texttt{commitValue}. The return value of the load (line~\ref{line:commit-load}) is the value returned by the  \texttt{commitValue} call. In this way, the process returns the same value from its load as any other process executing this load for this thunk.

To store a value in a given mutable variable, the process first executes a load as described above, thereby logging the value present before the store occurred, or discovering what that value was (if this store was already executed by a different process). The process then executes a CAS with old value equal to the value returned from the load. Recall that we assume that shared memory locations are ABA free, and therefore this ensures that all CAS attempts but the first will fail. The CAM operation works similarly to the store, but with an additional check to make sure the value returned by the load matches the expected value.
It only executes a CAS if this is the case.
By performing a load before the CAS, we guarantee that the expected value was stored in the memory location at some point.
Combined with the ABA-free assumption, this prevents a potentially dangerous scenario where the expected value is written into the memory location after the CAS, causing future executions of the CAM to no longer be idempotent.
It is important that the CAM does not return the return value of its CAS, since this value could be different for different processes that execute it, and could therefore violate idempotence (externalize a different result).

We also provide allocate and retire operations for idempotence. The idea is again to use the thunk log to commit values. To allocate a new object, the process allocates this object using the system-provided allocation mechanism, and then uses \texttt{commitValue} to install this new object in the log. If it is the first to do so, then the allocation is done, and this new object is returned. Otherwise, the process destroys its newly allocated object, and instead returns the object that was already installed in the log.

To retire an object, the processes use the log to compete for `ownership' of this object. The first process to commit a boolean retirement flag on the log is responsible for retiring this object. All other processes simply skip retiring it if they discover, by trying to commit a flag to the log, that some other process already owns this object. In this way, each object is retired at most once. Standard garbage collection techniques can then be used to collect retired objects when it is safe to do so.

The commitValue can also be used directly by the user to commit the result of any non-deterministic instruction.  For example, if there is an instruction that generates a value based on random noise in the processor,  this needs to be committed so all instances of the thunk agree on it.



\future{\Hao{low level operations (steps) consist of read, write, CAS, sysAlloc, sysRetire, sysFree. High level operations consist of load, store, CAM, allocate, free. val.CAS(), val.read()}}

We now show that our idempotence construction is correct; that is, \hedit{the mutable type implemented in Algorithm~\ref{alg:idempotent} is linearizable}, and any thunk that wraps all its mutable shared variables with the \texttt{mutable} type is idempotent.
\hedit{We begin by outlining a proof of idempotence.}
\hedit{For the following theorem, we relax Definition 1 so that retire operations in $E'$ are allowed to appear later than they would have in a single run of $T$.	
This has no effect on correctness and at worst it delays the reclamation of memory.
Our idempotence construction requires this relaxation because a process can go to sleep before peforming the \texttt{sysRetire} on line~\ref{line:retire}, and in the meantime, other processes can perform future operations of the thunk, making the retire appear out-of-order.
}
%

%


\future{\Hao{The alloc ensures that everyone has the same state after a big step. And the retire ensures that sysRetire is called exactly once.}}

\begin{theorem}
\label{theorem:idempotent}
	Replacing each \hedit{mutable shared} variable \hedit{accessed by} a thunk $T$ with a \texttt{mutable} and allocating and retiring all objects in $T$ with the provided allocate and retire operations yields an idempotent version of $T$.
\end{theorem}

\begin{proof}
(Outline.)  A complete proof which lines up with our relaxed definition of idempotence is given in the appendix.   Here we give a brief outline.   The idea is that all processes running the same thunk (descriptor) will stay synchronized in the sense they will have the same state at the same point of their execution.   Whichever gets to a loggable event first will log it, and all others will see it is already logged and use the same value.   In this way, they all see the same values, and stay synchronized.   It also means their position in the log will be synchronized.    Memory allocation and retiring is safe since only the first will keep its allocated value and only the first will retire the value.   For stores and CAMs, only the first such operation will succeed and all others will fail, because of our ABA-free assumption.  Therefore only the first will be visible.
\end{proof}


Idempotent by itself does not guarantee that we are not over-allocating or double freeing.
To prevent memory leaks, every block of memory allocated on line~\ref{line:allocate} that does not get committed to the log is freed on line~\ref{line:free}.
We also use the shared log to ensure that each object is retired no more than once on line~\ref{line:retire}.

\hedit{To complete the correctness proof, we also need to show that \texttt{load}, \texttt{store}, and \texttt{CAM} are linearizable in executions where each instance is run only once.
Intuitively, this is because in the absence of repeated runs, the \texttt{load} operation simply reads and returns the variable \texttt{val}, and the \texttt{store} and \texttt{CAM} operations simply read \texttt{val} and try to update it with a CAS.
This is a well-known linearizable implementation \texttt{load}, \texttt{store}, and \texttt{CAM}/\texttt{CAS} using just \texttt{load} and \texttt{CAS}, and it is linearizable as long as \texttt{store}s and \texttt{CAM}s do not race.
}

%

As mentioned, most previous approaches to idempotence have been based on \emph{context saving}~\cite{ben2019delay,blelloch2018parallel,turek1992locking,barnes1993method,lltheory}.  This involves storing out a program counter and current state of all local variables at important events (e.g. shared memory operations), and possibly loading and installing a new context if already stored.  Our approach never needs to store a program counter or local state since the processes are running ``synchronously'' and have the same local state.  For large thunks, and frequent helping, however, our method potentially does have an additional cost.  In particular, we always start helping from the beginning of a thunk while the other methods will start at the point of the last context saved by any process.  Our method is therefore particularly well suited for short thunks, which is the intended use with fine-grained locks, and possibly not as well suited for long running thunks.

\section{Lockless Locks}
\label{sec:locks}

\renewcommand{\figurename}{Algorithm}
\begin{figure}
\begin{lstlisting}
struct lockDescr :
   descriptor* d;
   bool isLocked;

type Lock = mutable<lockDescr>;

bool runAndUnlock(Lock* lock, lockDescr descr):
	bool result = run(descr.d);
	descr.d->done.store(true);
	lock->CAM(descr, lockDescr(descr.d, false));

bool tryLock(Lock* lock, Thunk @$f$@):
	bool result = false;
	lockDescr currentDescr = lock->load(); @\label{line:readlock}@
	if (not currentDescr.isLocked) : @\label{line:checklock}@
		descriptor* myDescr = createDescriptor(@$f$@); @\label{line:create}@
		lockDescr myLockedDescr = {myDescr, true}; @\label{line:tag}@
		lock->CAM(currentDescr, myLockedDescr);  @\label{line:attempt}@
		currentDescr = lock->load(); @\label{line:reload}@
		if ((myLockedDescr.d->done).load() or  @\label{line:checkdone}@
		     myLockedDescr  == currentDescr) :  @\label{line:rereadlock}@
			result = runAndUnlock(myLockedDescr);  //run self @\label{line:run1}@
		else if (currentDescr.isLocked) :
			runAndUnlock(currentDescr); // help other @\label{line:help1}@
		retireDescriptor(myDescr); @\label{line:retiredesc}@
	else : runAndUnlock(currentDescr); // help other @\label{line:help2}@
	return result;

void unlock(Lock* lock): @\label{line:unlock}@
 	lockDescr currDescr = lock->load();
	lock->CAM(currDescr, lockDescr(currDescr.d, false));
 \end{lstlisting}
\caption{Idempotent TryLock}
\label{alg:trylock}
\end{figure}

We now describe how we implement a \trylock{}.  It is important that \trylock{}s can be nested, which means the locking mechanism itself must be idempotent or otherwise safe to use when there are multiple threads helping to acquire the lock.  In particular, consider an operation $O_1$ that takes an outer lock $L_a$ and inside the lock takes an inner lock $L_b$.  If another operation $O_2$ encounters $L_a$ locked, it will help $O_1$ execute its critical code.  This means it will help $O_1$ acquire $L_b$ and, if successful, run the code of $O_1$ in that lock. 

Based on our technique for idempotence, it turns out to be quite simple to implement the locking mechanism.  In particular, we simply need to ensure that the code for the locking mechanism is itself idempotent, so that helping it is safe.  Our code is given in Algorithm~\ref{alg:trylock}.  A lock descriptor (\texttt{lockDescr}) is represented as a pair of a pointer to a descriptor and a boolean indicating whether it is currently locked or not.  It is easy to put these into a single word by stealing a bit from the pointer.   A \texttt{Lock} is then a \texttt{mutable} lock descriptor.

An attempt at acquiring the lock starts by reading the lock and checking if it is currently locked.  If not locked, the algorithm creates a descriptor for the thunk $f$ (line~\ref{line:create}) and tags it to mark that it is locked (line~\ref{line:tag}).  It then attempts to install the descriptor on the lock using a CAM (we do not have a CAS for mutables).  Since the CAM does not return whether it succeeds, the algorithm needs to read the lock again (line~\ref{line:rereadlock}) to check if successfully acquired.  If acquired
or if previously acquired and now done, it runs the code and releases the lock (line~\ref{line:run1}).  
If not acquired but \texttt{currentDescr} is locked, then the algorithm helps the descriptor on the lock and unlocks
it (line~\ref{line:help1}).
Whether the CAM was successful or not, \texttt{myDescr} needs to be retired (line~\ref{line:retiredesc}).  If on line~\ref{line:checklock} the lock is already locked, then the algorithm helps the descriptor on the lock
and unlocks it (line~\ref{line:help2}).  Finally the \texttt{result} is returned, which will only be true if the lock was successfully acquired and the thunk $f$ returns true.

We now argue correctness.  We say a \trylock{} is correct if it either fails, in which case none of the critical code (thunk $f$) is run and it returns false; or it succeeds, in which case all its critical code is run and the \trylock{} returns its value.
If successful, no \hedit{other} critical code on the same lock can run concurrently.  By this definition, the \trylock{} could always fail, but this would not satisfy progress bounds, and in particular for us, our lock-free bounds.
We say a successful \trylock{} \emph{enters} on the step the lock is changed to point to its descriptor and \emph{exits} on the step when the lock is changed from locked with its descriptor to unlocked.    

\begin{theorem}
\label{theorem:trylock}
  The \trylock{} in Algorithm~\ref{alg:trylock} is correct as long as \texttt{run(descriptor)} runs the user code in the thunk $f$ idempotently, and the operations on a Lock (\texttt{load}, \texttt{CAM} and \texttt{store}) and on descriptors (\texttt{createDescriptor} and \texttt{retireDescriptor}) are idempotent.
\end{theorem}
\begin{proof}
  (Outline).  The code in a thunk consists of the user level code and possibly the code of a nested \trylock{}.  Together this is idempotent by assumption.  
  
  In the algorithm, a descriptor is run if and only if the \trylock{} enters and the lock is set.
  The descriptor is run by the \texttt{runAndUnlock} method which can be called on line~\ref{line:run1} by the process that installed the descriptor, or on lines~\ref{line:help1} or~\ref{line:help2} by the helping processes. 
  Some process will finish the thunk first (either the primary process or a helper).  Since the thunk is idempotent, any processes working on the same descriptor after that point will have no effect.  The lock is only released after the thunk is first finished so the code can only have an effect between when the successful
\trylock{} enters and exits.  Since there is a unique descriptor on the lock during this time, no other thunk on the same lock can appear to run concurrently (there could be leftover thunks from earlier successful attempts on the lock, but they will have no effects).

 If either the lock was already taken on line~\ref{line:readlock} (i.e. the check on line~\ref{line:checklock} fails) or the attempt to install a descriptor was unsuccessful on line~\ref{line:attempt} (i.e. the check on line~\ref{line:rereadlock} fails), then the \trylock{} fails and returns false.
 Otherwise, its descriptor was successfully installed, and it returns the result of running that descriptor on line~\ref{line:run1}.
 Note that it is important to check the descriptor's done flag on line~\ref{line:checkdone} because even when the descriptor is successfully installed on line~\ref{line:attempt}, the \texttt{load} on line~\ref{line:reload} might not see it because it might have been helped and replaced by another process.
 Checking the done flag ensures that the \trylock{} will always return the return value of the descriptor it installed if its CAM on line~\ref{line:attempt} is successful.
 \end{proof}
The theorem does not depend on a particular implementation of idempotence, but works with ours since ours satisfies
the specified conditions.

We now show that \trylock{}s are lock-free.  For this purpose we make some assumptions.  Firstly, we assume the locks have a partial order $<_p$, and that when nesting locks they are acquired in decreasing order.  This is a relatively standard assumption for lock-based algorithms since it prevents lock-cycles and deadlock.  In lists and trees, the ordering can be implied by the ordering in the tree.  Secondly, we assume that each \trylock{} includes at most one other \trylock{} directly inside of it.  Note that this still allows arbitrary depth of nesting since the one inside can itself contain another lock inside it.\footnote{We expect this requirement is not necessary, but our proof relies on it and it is true for all our \trylock{}-based data structures.}  We refer to locks that satisfy these two conditions as \emph{simply nested}.  As with some other lock-free mechanisms~\cite{shavit1997software,fraser2007concurrent} we also assume the number of locks (or in their case, memory locations) is bounded.
We say that a simply nested \trylock{} \emph{\recsucceed{}s} if it acquires its lock as well as all locks nested inside of it.  Note that for a \trylock{} 
$a$ if any one \trylock{} nested in it \recsucceed{}s then they all do, including $a$.

We also need to bound the time of user code in a lock, otherwise helpers could never complete helping.  We defined \emph{step count} for a \trylock{} as the number of user steps taken by its critical region.  We count all functions in the idempotent interface as unit cost plus \hedit{the cost of} any user code inside of them---in particular \texttt{sysAllocate} and \texttt{sysRetire} count toward user code.  We count a nested \texttt{try\_lock} as unit cost plus the step count of the code in its critical region.


\begin{theorem}
\label{theorem:lockfree}
Consider an algorithm using simply nested \trylock{}s for which the maximum step count for any \trylock{}, not including helping, is bounded.   In such an algorithm, a \trylock{}, including any helping it does, will run in bounded steps, and for every bounded
number of \trylock{} attempts at least one top-level \trylock{} will \recsucceed{}.
\end{theorem}

\begin{proof}
We say a \trylock{} $a$ with a nested (not necessarily directly nested) \trylock{} $b$ \emph{helps} another \trylock{} $c$ if (1) line~\ref{line:help1} or line~\ref{line:help2} of $b$ runs the thunk installed by $c$
or (2) \texttt{currentDescr} on line~\ref{line:reload}  of $b$ is unlocked and belongs to $c$.  In the second
case, $a$ does not actually help $c$, but $c$ must have acquired its lock after
line~\ref{line:readlock} of $b$ and released it before line~\ref{line:reload} of $b$ so we can give $a$ credit for helping. 
Importantly, by this definition, if a \trylock{} performs no helping, then it \recsucceed{}s.
Due to idempotence and simple nesting of locks, every \trylock{} helps at most one other. 

Now note that if $a$ helps $b$ due a conflict on lock $l_1$ and $b$ helps $c$ due to a conflict on lock $l_2$ the \trylock{} that attempts $l_2$ is nested inside the \trylock{} that acquired $l_1$.  Since simply nested locks are acquired in the partial order $>_p$, we have that $l_1 >_p l_2$ and more generally that locks decrease along any chain of helping.  Assuming a bounded number of locks, the chain will have bounded length and end with a \recsuccessful{}
\trylock{}.  Hence, running any \trylock{} takes bounded steps (including helping).  Furthermore, since there are a bounded number of locks on the chain, the number of \trylock{}s responsible for completing the last one is also bounded.  
Finally we note that although the last one might not be top-level, the fact it recursively succeeds implies the top-level \trylock{}
that contains it recursively succeeds.
 \end{proof}

This theorem indicates that simply nested \trylock{}s are lock-free in that a \trylock{} must succeed in a finite number of steps. 
It does not, however, imply wait-freedom since a particular process could continuously fail to acquire a lock.  It also does not, by itself, guarantee an algorithm using simply nested \trylock{}s is lock-free.  In an algorithm based on optimistic fine-grained locks, for example, we might need to retry not because a lock failed to be acquired, but instead because some consistency check failed (e.g. the test on line~\ref{line:removetest} of Algorithm~\ref{alg:list}).  In all the algorithms we consider, however, a consistency check can only fail if the algorithm has made progress.  In the \texttt{remove} from Algorithm~\ref{alg:list}, for example, the consistency check can only fail if in between the \texttt{find\_location} and when the lock on \texttt{cur} is acquired, either (1) \texttt{cur} is deleted or (2) it is updated to point to a new \texttt{next}.  In either case, the algorithm has made progress by completing an operation.  A similar argument can be made for the \texttt{insert}. 
Therefore the ordered list algorithm based on our \trylock{}s is lock-free, as are the other algorithms we consider.


It can be useful to release a lock early before the scope of the thunk associated with the acquired lock completes.   We supply a \texttt{unlock} for this purpose.   It takes a lock that is currently acquired by the thread and unlocks it.   Its behavior is undefined if the thread has not acquired the lock.   As mentioned in the introduction, this can be used for hand-over-hand locking (also called lock-coupling)~\cite{Bayer88}.

The code for \trylock{} can be modified to support a strictLock that always acquires the lock before returning, by first creating the descriptor, and then putting the attempt to acquire a lock into a while loop.  We have implemented an optimized version of such a strictLock and compare it to the \trylock{} in Section~\ref{sec:exp}.
We note that this implementation of strict locks is not simply nested so is not covered by Theorem~\ref{theorem:trylock}.  However, it should be possible to adapt TSP's proof~\cite{turek1992locking} to show that strict locks are lock-free.

	\section{Related Work}

As mentioned, the idea of lock-free locks was introduced by Turek et al~\cite{turek1992locking} and Barnes~\cite{barnes1993method}.  The idea of helping dates back earlier, at least to Herlihy's work on wait-free simulations~\cite{waitfree91}.  Many wait-free and lock-free algorithms achieve their progress guarantees by allowing processes to safely \emph{help} each other complete their operations, although in quite specific ways instead of using a general mechanism. Help used for wait-free progress was formally studied by Censor et al.~\cite{censor2015help}.

The idea of \emph{idempotence} has been used in the literature a variety of contexts~\cite{idempotence02,idempotence12,idempotence13,blelloch2018parallel,lucia2015simpler,KOEFV06,blelloch2018parallel,ben2019delay}.  Kruijf, Sankaralingam and Jha~\cite{idempotence12} give a nice overview although only up to 2012.   More
recent work has focused on using idempotence for fault tolerance (e.g., \cite{lucia2015simpler,blelloch2018parallel,ben2019delay}).
All these approaches rely on some form of  ``context saving''.
Idempotence has also been considered and characterized in the literature under different names. Timnat and Petrank~\cite{timnat2014practical} define a similar notion known as \emph{parallelizable} code, which intuitively allows several processes to execute it without changing its effects.   

In recent work, Ben-David and Blelloch in~\cite{lltheory} use a randomized implementation of lock-free locks to show that when point contention on locks in constant, then operations can be completed in constant expected time.  We use their definition of idempotence in this paper.  However, their focus is on theoretical efficiency and fairness guarantees of acquiring the locks, whereas in this paper we focus on the practicality of the approach.  As with previous approaches to idempotence their approach relies on context saving.


Approaches for achieving idempotence and lock-freedom sit on a spectrum of generality. The focus of this paper is to improve the practicality of the far side of the spectrum; fully general idempotence/lock-free constructions. However, many other approaches exist, which are less general but can be more efficient for their specific applications. For example, on the other end of the spectrum are hand-designed lock-free data structures. These data structures are often designed to be able to have `critical sections' that contain just one CAS instruction, and can therefore be executed atomically in hardware with no locks. For example, Michael and Scott's queue~\cite{michael1996simple} allows new nodes to be enqueued by swinging a single pointer. Idempotent help is given by later updating the tail pointer. Similar algorithms, like  Harris's linked-list~\cite{harris2001pragmatic} and Natarajan and Mital's BST~\cite{Natarajan14}, make use of descriptors to allow others to help, but these descriptors are optimized to simply be flags. These approaches yield very fast lock-free data structures, but are difficult to generalize. 

A middle-ground between generality and efficiency is found with approaches that implement useful primitives for lock-freedom. For example, Brown et al~\cite{BrownER13} introduce the LLX/SCX primitive, which allows atomically checking that several locations have not changed their values, `freezing' some of them, and modifying one of them. This primitive can be seen as a lock with a restricted critical section. Another example of such a primitive is multi-word CAS, which allows several memory locations to be CASed atomically~\cite{guerraoui2020efficient,feldman2015wait,harris2002practical}. 

Some work aims at achieving practical lock-free locks but only partially solve the problem.
Rajwar and Goodman describe a hardware-based technique that are lock-free under an assumption that processes do not fail or stall during certain critical regions~\cite{RajwarG02}.   We assume a process can fail or stall at any instruction.
Gidenstam and Papatriantafilou~\cite{GidenstamP07} look at how to make the handoff of locks lock-free (i.e., waking up threads suspended on a lock in a lock-free manner), but a thread blocked during a lock will still delay any waiting threads indefinitely.


	\section{The \flock{} Library}
\label{sec:imp}

We have implemented a C++ library, \flock{}, based on our \locklesslock{}s approach.  It supports a \texttt{mutable\_} wrapper to use on any shared values that can be mutated inside a lock, a \texttt{lock} type and a \texttt{try\_lock}.  The \texttt{mutable\_} wrapper has a similar interface to the C++ \texttt{atomic} wrapper.  In particular, it supports \texttt{load}, \texttt{store} and \texttt{cam}.  The assignment operator (\texttt{=}) is overloaded to store.  \flock{} also supports \texttt{allocate} and \texttt{retire} which are integrated with its epoch-based collector.  An example of how to use \flock{} is given in Algorithm~\ref{alg:list}.
The library is available at \url{https://github.com/cmuparlay/flock}.

Here we discuss several specifics about the implementation, including some optimizations.

\myparagraph{Epoch-based collection} 
\flock{} uses an epoch-based memory manager~\cite{fraser2004practical,mckenney2008rcu}.  In such a memory manager, each operation runs in an epoch, each of which is associated with an integer that increases over time.  Managing memory with epochs requires some additions to the the implementation of idempotent code.  In particular, when a thread helps another thread, it is taking on the responsibility of that other thread.  It therefore needs to also take on its epoch number.  
To implement this, when \flock{} has to help inside of a \texttt{try\_lock}, it changes its epoch to be the minimum of its epoch and the epoch of the thunk it is helping.  When it is finished helping, it restores its epoch to what it was before helping.
The descriptors are also allocated and retired with the same epoch-based collector, with one optimization.  In particular if
a descriptor is never helped, which is the common case, then it can be reused immediately instead of being retired.   To implement this, we keep a flag on the descriptors which is set when helping.  This requires some careful synchronization.

\myparagraph{ABA}
Although the idempotent implementation in Algorithm~\ref{alg:idempotent} requires that mutables are ABA free, a \texttt{mutable\_} in \flock{} does not have this requirement.  To allow for this, \flock{} keeps tags on mutable locations.  A simple implementation is to use a 64-bit counter, and increment the counter on each update.  Assuming mutable values can be up to 64-bits, this can be implemented with double-word (128-bit) loads and CASes.  Unfortunately double-word loads are particularly expensive on current machines.  
\flock{} has two optimizations to avoid them, one which supports 64-bit values, and one for 48-bit values, which is sufficient for a pointer.  

The first optimization still uses a 64-bit counter on every mutable, but avoids any double-word loads.  A key observation is that a \texttt{load} only needs to log the value, and therefore only needs to read this value.  Another observation is that a \texttt{store} (or \texttt{cam}) does not need to read the counter and value atomically.  Instead, it can first read the counter and then the value, followed by a double-word CAS to the mutable.  This is safe since the value can only change if the counter changes.

The second optimization avoids the extra 64-bit counter on each mutable location and any double word operations.  Instead it uses a safe lock-free approach that only requires a 16-bit tag.  The full description is beyond the scope of this paper, but roughly it uses an announcement array to ensure that wrapping around is safe---i.e., it never uses a tag that is announced.  All the experiments in Section~\ref{sec:exp} use this version since the mutables are no larger than a pointer.

\myparagraph{Constants and Update-once Locations}  Shared, constant locations do not need to be wrapped in a \texttt{mutable\_} and can just be read directly.    A constant location is one that is written once and only read after it is written.  The write could happen during construction of the object that contains it or after.   For example the key and value in the list link in
Algorithm~\ref{alg:list} are constants.   \flock{} also supports update-once locations.  These are locations that have an initial value, and are updated at most once.  Reads can happen before or after the update.    The \texttt{removed} flag in a link in
Algorithm~\ref{alg:list}, for example, is updated once.    Update-once variables are ABA free and therefore do not need a tag.  Furthermore, the \texttt{store} can be implemented with a simple write instead of a \texttt{load} and then a \texttt{CAS}.  This is because only the first such write will have an effect.

\myparagraph{Arbitrary Length Logs} In general it cannot be determined ahead of time how long a log will be.  \flock{} therefore implements logs that can dynamically increase in size.  In the implementation, a log has a fixed block size (7 by default).  If it runs out, another block is allocated.  To do this idempotently, the first thread that runs out allocates the block and attempts to CAS it into a next-block pointer.  If it fails, it frees its block and takes on the block that succeeded.

\myparagraph{Avoiding CASes} We found that one of the most expensive aspects of helping is contention due to CASes on both the log and mutable locations.  This is especially true under high contention when there is a lot of helping.  To significantly reduce this contention we use a compare-and-compare-and-swap.  In particular, before doing a CAS, the location is read and compared against the expected value, and if not equal the CAS can be avoided.  When helping under high contention it is often not equal (someone else already executed the CAS) so many of the CASes are avoided.  This rather simple change made a significant improvement in performance under high contention---sometimes a factor of two or more.

\myparagraph{Capturing by Value}  In the code in Algorithm~\ref{alg:list}, one might notice the ``\texttt{[=]}'' in the definition of the lambda's.  This indicates that all free variables in the lambda defined outside of it are captured by value, as opposed to by reference---i.e., they are copied into the thunk.  This is important since the lambda might outlive its context, and any surrounding stack allocated values could be destructed while being helped.  Indeed if the \texttt{[=]} is replaced by \texttt{[\&]} (by reference), Algorithm~\ref{alg:list} would be incorrect---for example, the variable \texttt{prev} on line~\ref{line:loadRemove} could be reused while the lambda is being helped.   

\section{Data Structures}
\label{sec:ds}

We have implemented several concurrent data structures using \flock{}.
These data structures include the doubly linked list described in Section~\ref{sec:listexample} (\dlist), a singly-linked list~\cite{lazylist06} (\lazylist{}), an adaptive radix tree~\cite{art13,arttre16} (\arttree{}) which is a state-of-the-art index data structure used in the database community, a separate chaining hashtable (\hashtable{}), a leaf-oriented unbalanced BST (\leaftree{}), a leaf-oriented balanced BST (\leaftreap{}) with an optimization that stores a batch of key-value pairs (up to 2 cachelines worth) in each leaf to minimize height, and an (a,b)-tree (\abtree{}).
To support concurrent accesses, the data structures use fine-grained, optimistic locking, as in ~\cite{KungL80,arttre16,lazylist06,bronson10}.
This approach involves (1) traversing the data structure without any locks, (2) locking a neighborhood around the nodes you wish to modify, (3) checking for consistency, and (4) performing the desired modifications.
If the consistency check fails, locks are released and the operation restarts.
Read-only operations do not take any locks.

We implement a \trylock{} and a \strictlock{} version of each data structure. Both \trylock{} and \strictlock{} can either be lock-free (with helping and logging) or blocking (using test-and-test-and-set locks), and with our library, this choice can be made by changing a flag at runtime.

To the best of our knowledge, this results in the first lock-free implementation of an adaptive radix tree. In many workloads, our lock-free \arttree{} significantly outperforms the other lock-free ordered set data structures that we ran.
Our implementations of these optimistic, fine-grained locking data structures are available at \url{https://github.com/cmuparlay/flock}.


\renewcommand{\figurename}{Figure}

\section{Experimental Evaluation}
\label{sec:exp}

Our experimental evaluation has two main goals: first, to compare the performance of \locklesslock{}s with \blocklock{}s 
and second, to compare data structures written with \locklesslock{}s with state-of-the-art alternatives.


\myparagraph{Setup}
Our experiments ran on a 72-core Dell R930 with
4x Intel(R) Xeon(R) E7-8867 v4 (18 cores, 2.4GHz and 45MB L3 cache),
and 1Tbyte memory. 
Each core is 2-way hyperthreaded giving 144
hyperthreads. 
The machine's interconnection layout is
fully connected so all four sockets are equidistant from each other.
We interleaved memory across sockets using \texttt{numactl -i all}.
The machine runs Ubuntu 16.04.6 LTS.
We compiled using g++ 9.2.1 with \texttt{-O3}.
We used ParlayLib~\cite{blelloch2020parlaylib} for scalable memory allocation.
\Guy{I think I took the padding out.  It did not seem to make a difference.  Could possibly drop this sentence.}
\Hao{dropped last sentence}


\myparagraph{Workloads}
We experiment with set data structures supporting \op{insert}, \op{delete} and \op{lookup} with 8-byte keys and 8-byte values.
Our experiments follow a similar methodology to previous papers~\cite{DGT15,Brown18}.
We first pick a key range $[1, r]$ and prefill the data structure with half the keys in the range.
Then each thread performs a mix of \op{lookup} and update operations, where update operations are evenly split between \op{insert}s and \op{delete}s, keeping the data structure size stable throughout the run.
Each experiment is run for 3 seconds (sufficient for reaching a stable state) and repeated 4 times.
The first run is a warmup run and an average of the last 3 runs is reported.
Before the warmup run, we shuffle the ParlayLib memory allocator by allocating a large number of nodes and freeing them in a random order to increase consistency across runs.
\hedit{Standard deviation between runs is small enough that the error bars in our graphs are only visible for a small number of data points.}
\Guy{Do we have error bars?   Do you mean ``would not be visible''?}
\Hao{We do have error bars. Reworded the previous sentence to mention this.}

All keys are randomly chosen from the range $[1, r]$ according to a zipfian distribution parameterized by $\alpha$.
Zipfian with $\alpha = 0$ is identical to the uniform distribution and 
higher $\alpha$ skews accesses towards certain ``hot'' keys, which is more representative of real-world workloads.
The zipfian distribution is also used in the YCSB benchmark suite, which mimics OLTP index workloads~\cite{cooper2010benchmarking}.
We mostly run with 5\% and 50\% updates, following YCSB Workloads B and A, respectively.

Our experiments vary four parameters: data structure size, update rate, $\alpha$, and number of threads.
We show graphs along each of the dimensions, fixing the other three. Since \arttree{} is a trie data structure, it benefits heavily from densely packed keys, so we sparsify the key range by hashing each key from $[1, r]$ to a 64-bit integer. This does not affect the other data structures since they either are purely comparison based or hash the keys themselves.

\begin{figure}
	\centering
	\includegraphics[width=0.42\textwidth,trim={0 0.2cm 0 0},clip]{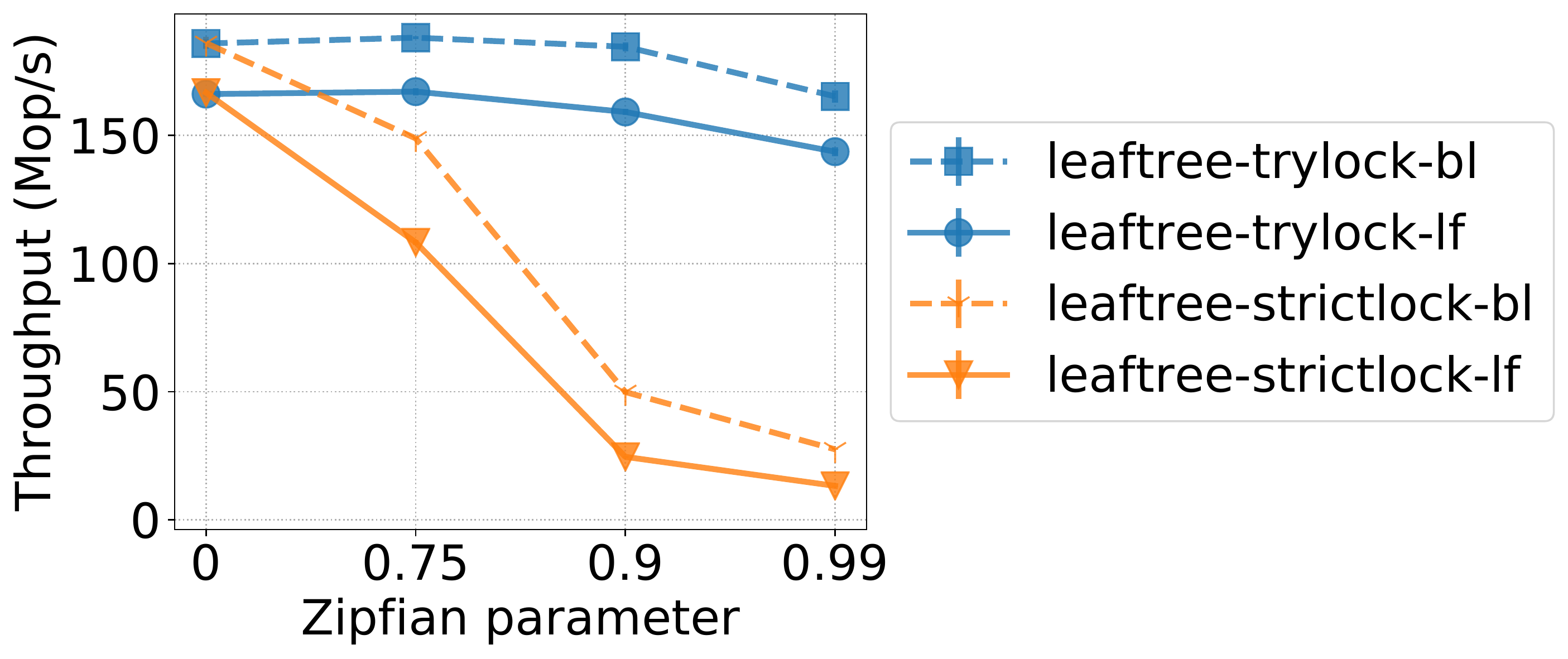}
	\caption{Comparing \tlock{} with \slock{} on workload with 100K keys, 144 threads, and 50\%
		updates. The `bl' and `lf' suffixes represent the blocking and lock-free version of our locks, respectively.}
	\label{fig:try-vs-strict}
\end{figure}

\myparagraph{Try vs strict lock} In data structures that employ optimistic locking, \trylock{} is often preferable to \strictlock{}.
This is because optimistic locking requires checking for consistency after taking the necessary locks.
So if a process $p_1$ tries to acquire a lock that is held by a another process $p_2$, it is better for $p_1$ to restart its operation instead of waiting to acquire the lock because it will likely fail its consistency check due to modifications by $p_2$.
We see this happen in the \leaftree{} in Figure~\ref{fig:try-vs-strict}.
The higher $\alpha$ is, the more contention there is on the locks, and the more beneficial \trylock{} becomes.
This holds for both \blocklock{}s and \locklesslock{}s.
In the rest of this section, we only report on \trylock{}s. 


\begin{figure*}
	
	
\begin{subfigure}{0.99\textwidth}
\includegraphics[width=0.69\textwidth,trim={0cm 0 0 0},clip]{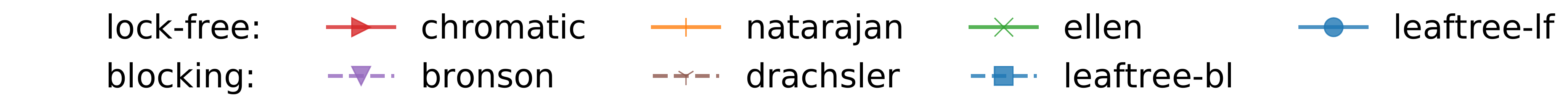}
\end{subfigure}
%

%

\centering


	\begin{subfigure}{0.24\textwidth}
		\centering
		\includegraphics[width=\textwidth]{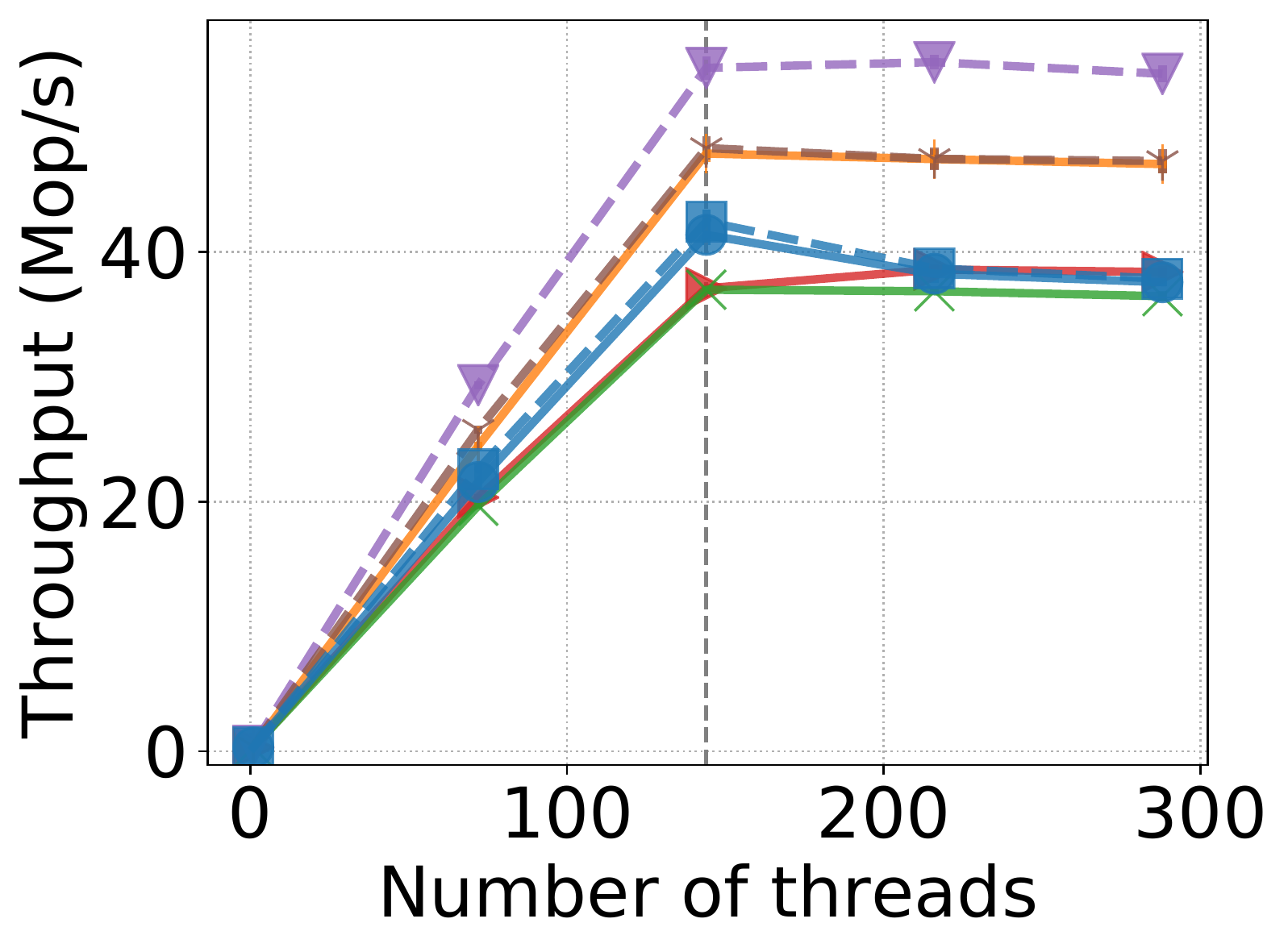}
		\captionsetup{justification=centering}
		\caption{100M keys, 50\% up., $\alpha = 0.75$}\label{subfig:tree-100M-threads}
	\end{subfigure}
	\begin{subfigure}{0.24\textwidth}
		\centering
		\includegraphics[width=\textwidth]{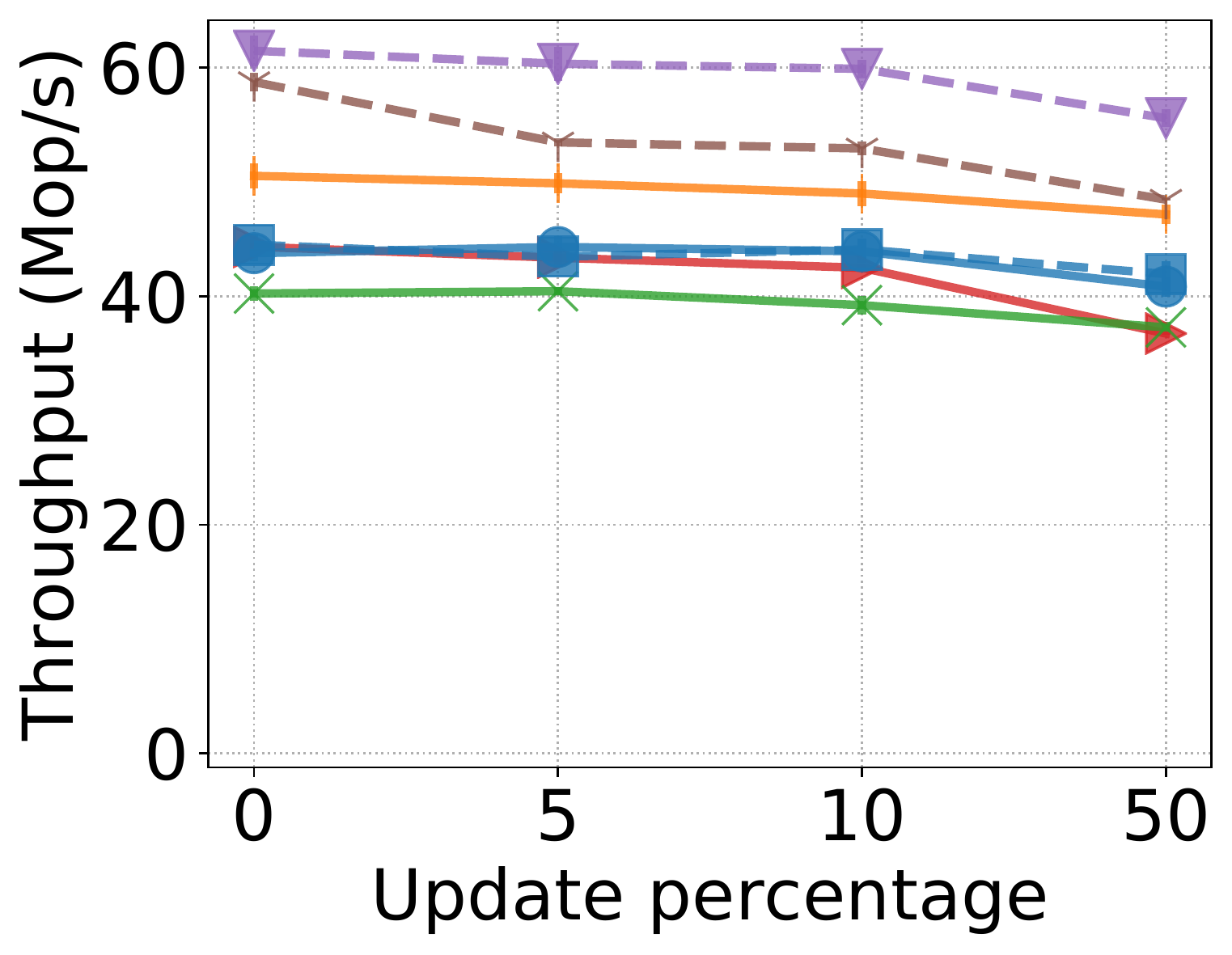}
		\captionsetup{justification=centering}
		\caption{100M keys, 144 th., $\alpha = 0.75$}\label{subfig:tree-100M-updates}
	\end{subfigure}
	\begin{subfigure}{0.24\textwidth}
		\centering
		\includegraphics[width=\textwidth]{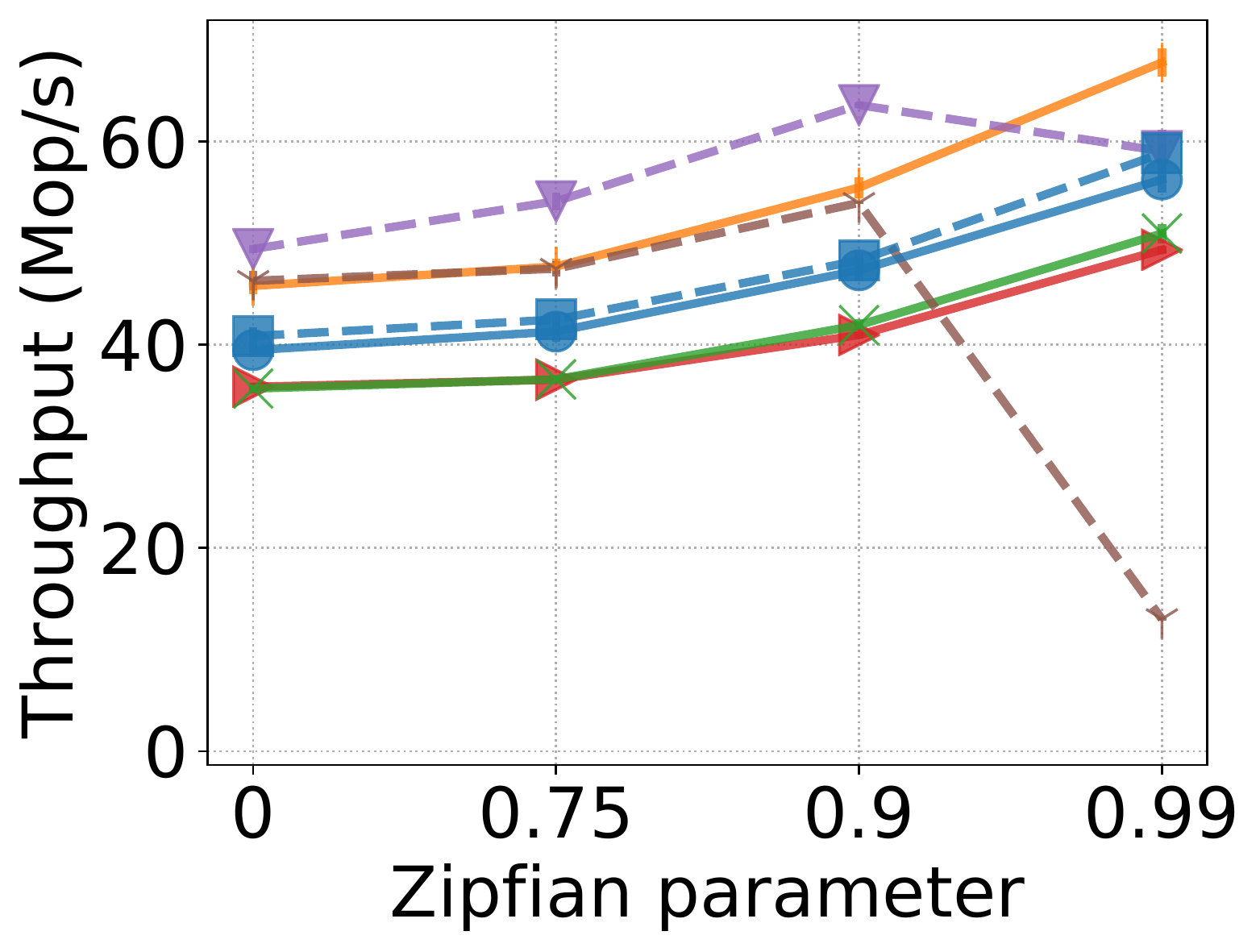}
		\captionsetup{justification=centering}
		\caption{100M keys, 144 th., 50\% up.}\label{subfig:tree-100M-zipf-144}
	\end{subfigure}
	\begin{subfigure}{0.24\textwidth}
		\centering
		\includegraphics[width=\textwidth]{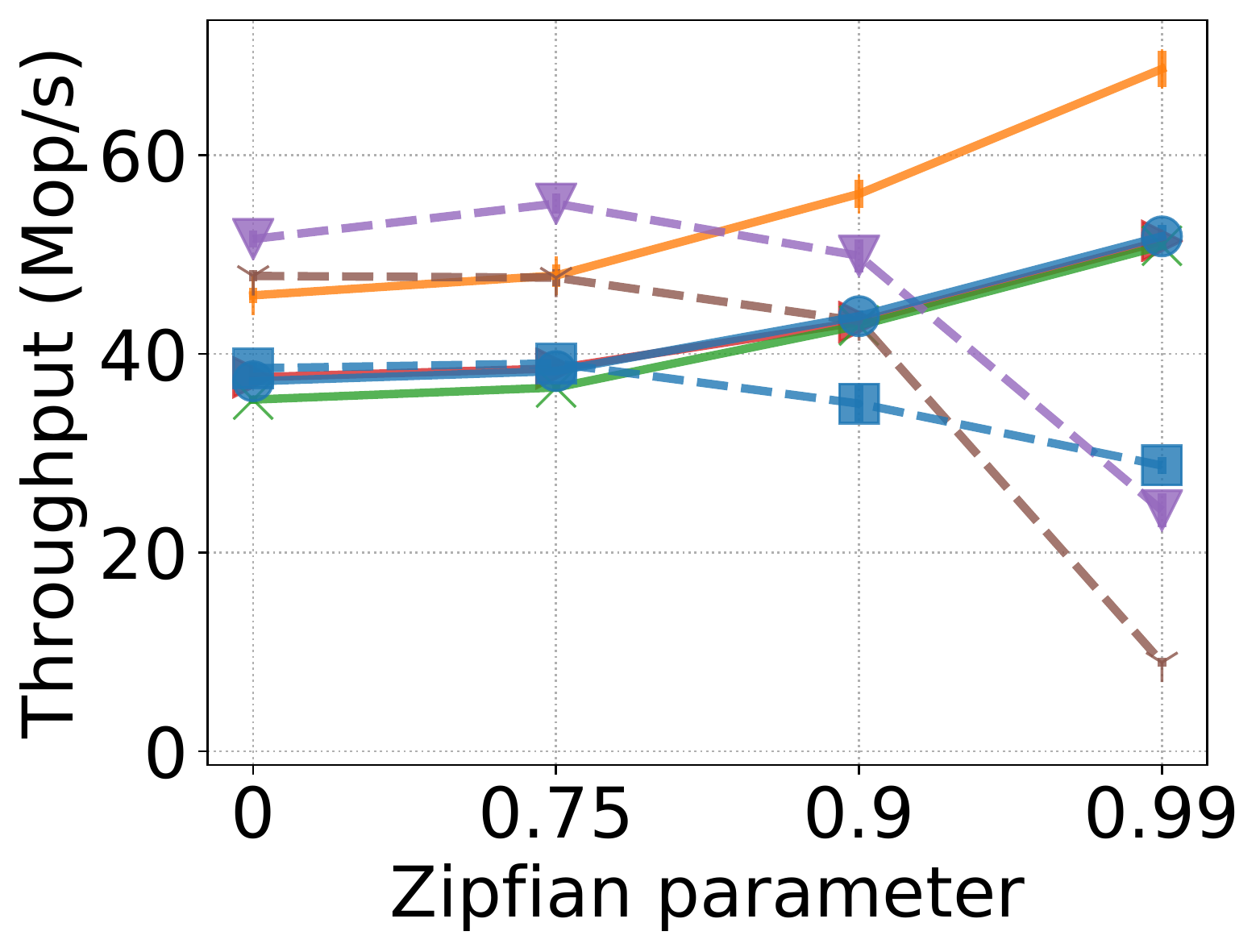}
		\captionsetup{justification=centering}
		\caption{100M keys, 216 th., 50\% up.}\label{subfig:tree-100M-zipf-200}
	\end{subfigure}
	
	\bigskip
	
	\begin{subfigure}{0.24\textwidth}
		\centering
		\includegraphics[width=\textwidth]{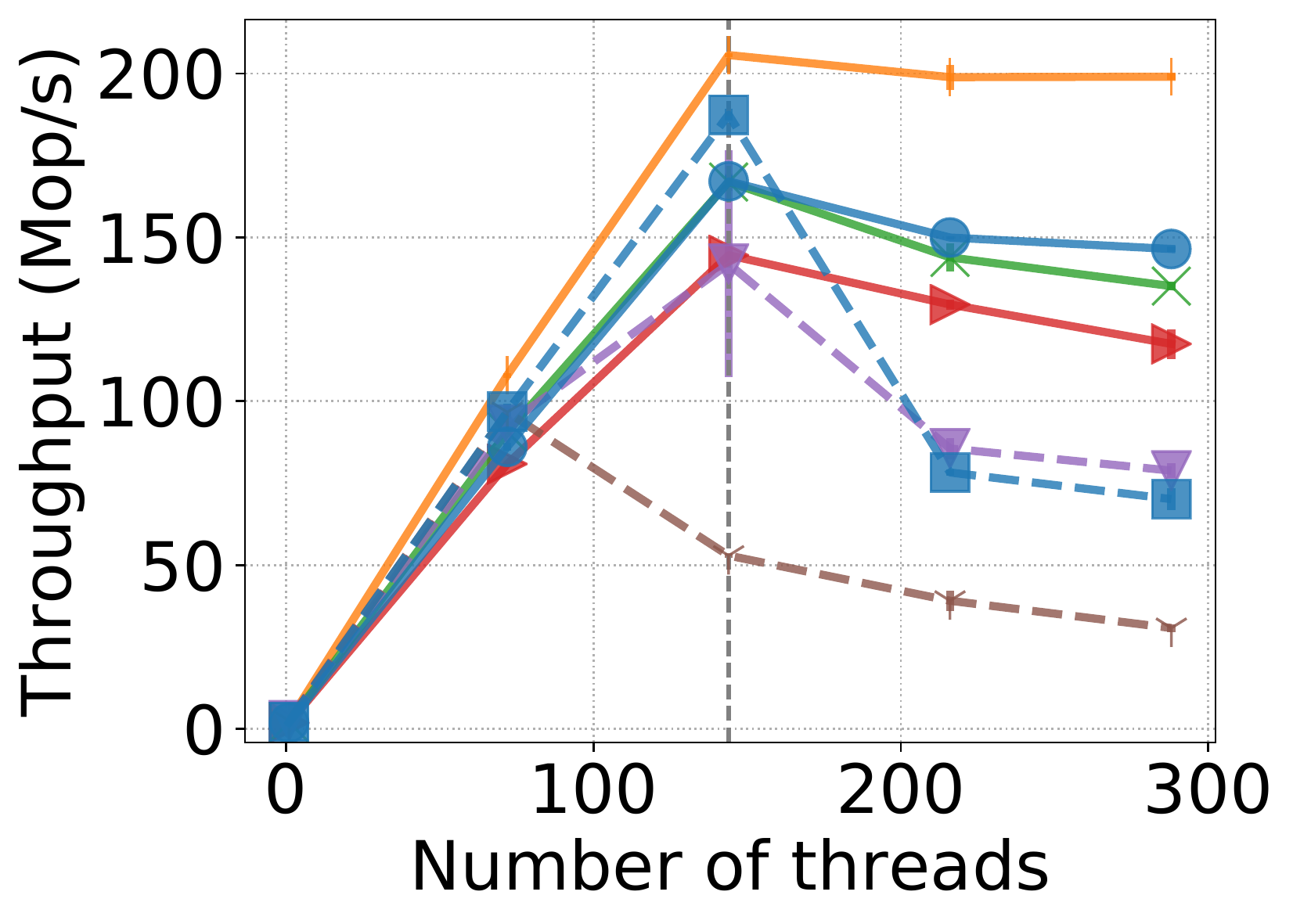}
		\captionsetup{justification=centering}
		\caption{100K keys, 50\% up., $\alpha = 0.75$}\label{subfig:tree-100K-threads}
	\end{subfigure}
	\begin{subfigure}{0.24\textwidth}
		\centering
		\includegraphics[width=\textwidth]{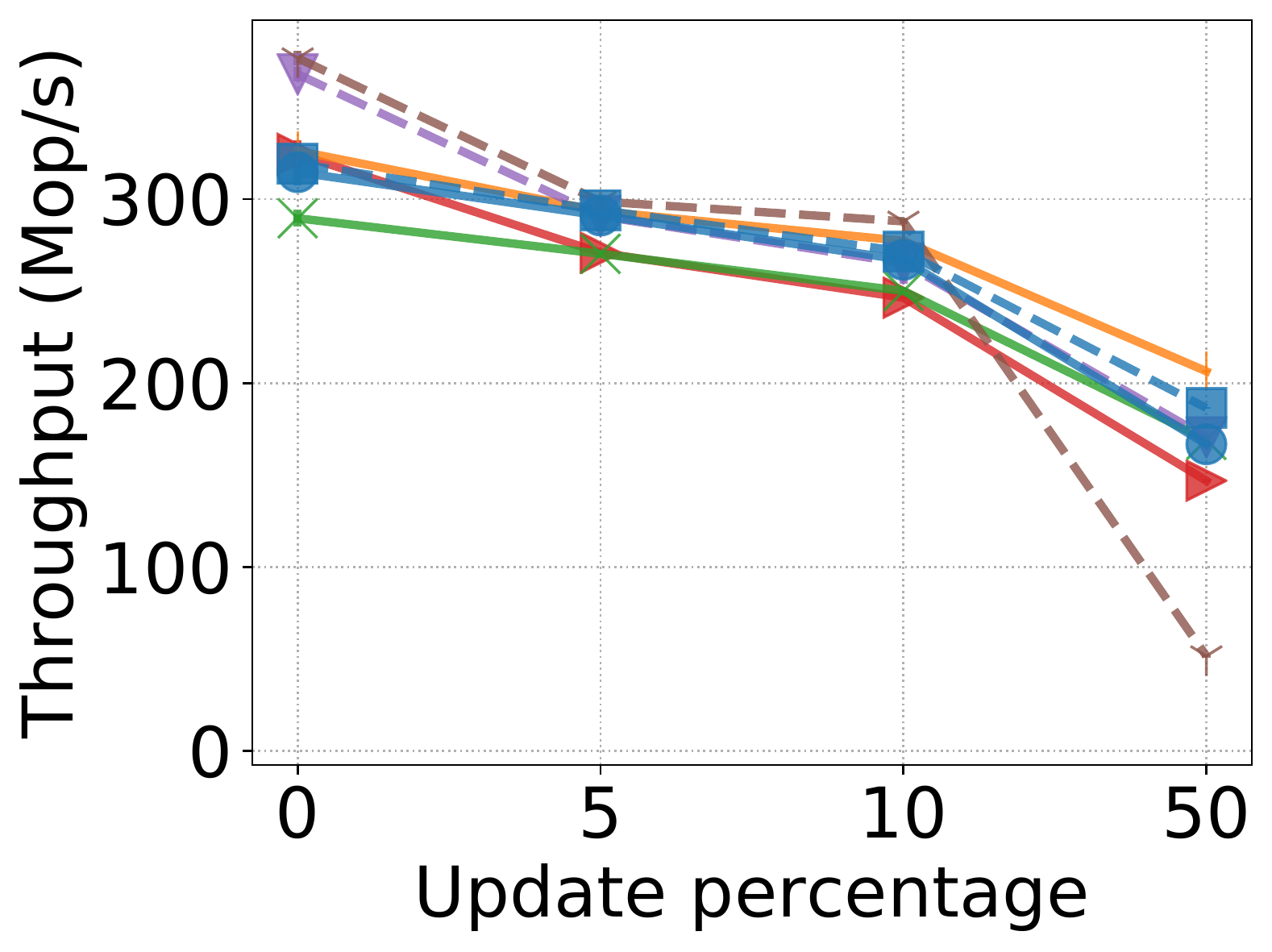}
		\captionsetup{justification=centering}
		\caption{100K keys, 144 th., $\alpha = 0.75$}\label{subfig:tree-100K-updates}
	\end{subfigure}
	\begin{subfigure}{0.24\textwidth}
		\centering
		\includegraphics[width=\textwidth]{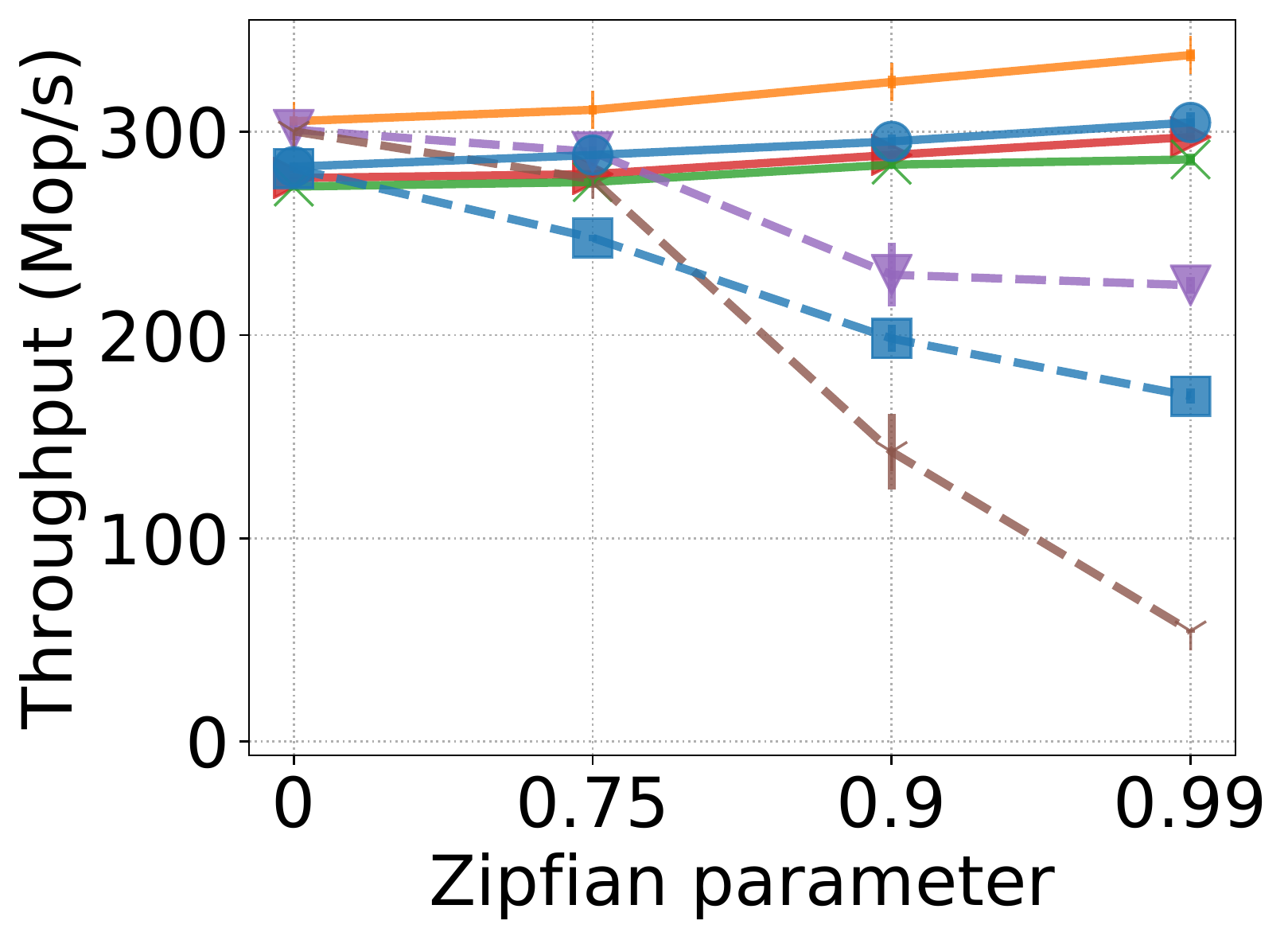}
		\captionsetup{justification=centering}
		\caption{100K keys, 216 th., 5\% up.}\label{subfig:tree-100K-zipf}
	\end{subfigure}  
	\begin{subfigure}{0.24\textwidth}
		\centering
		\includegraphics[width=\textwidth]{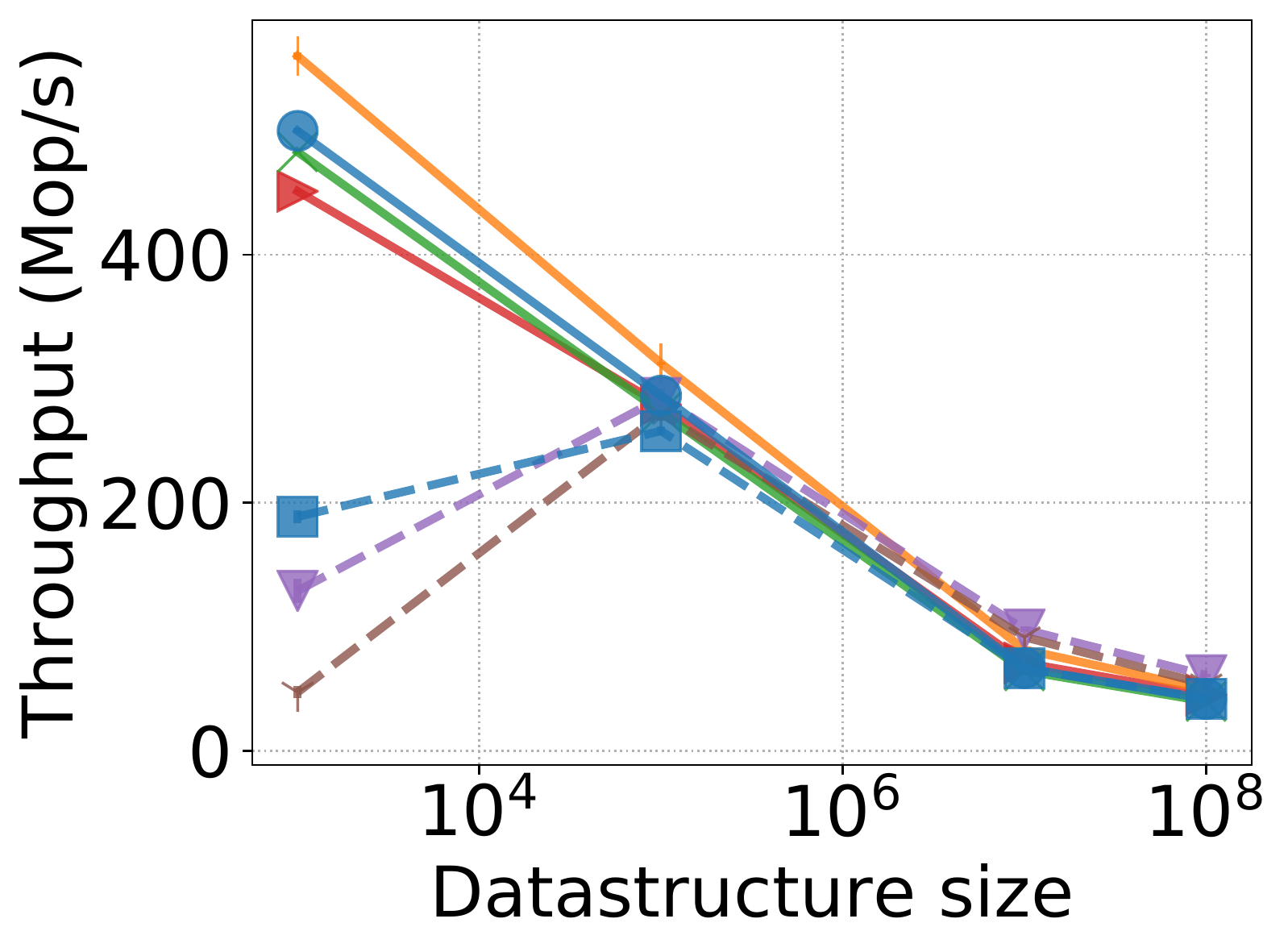}
		\captionsetup{justification=centering}
		\caption{216 th., 5\% up., $\alpha = 0.75$}\label{subfig:tree-sizes}
	\end{subfigure}
	\caption{Throughput of binary trees under a variety of workloads are shown. Dotted lines are used for blocking data structures and solid lines for lock-free ones. Subcaptions abbreviate `threads' to `th' and `updates' to `up'. The `bl' and `lf' suffixes represent the blocking and lock-free version of our locks, respectively.
	}
	\label{fig:tree-exp}
\end{figure*}

\begin{figure}
	\centering
	\begin{subfigure}{0.49\textwidth}
		\centering
		\includegraphics[width=\textwidth,trim={4.1cm 0 0 0},clip]{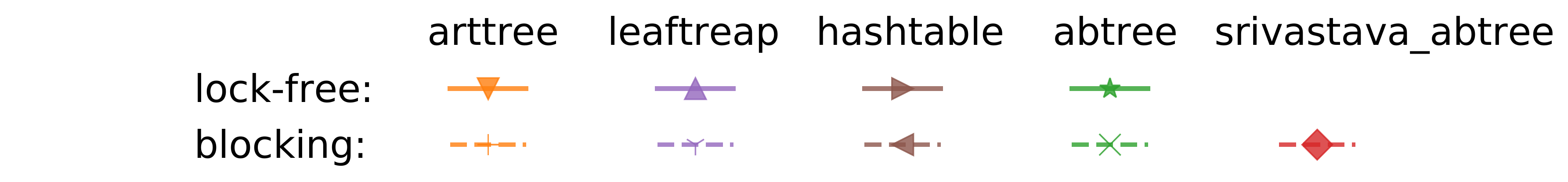}
	\end{subfigure}
	\begin{subfigure}{0.23\textwidth}
		\centering
		\includegraphics[width=\textwidth]{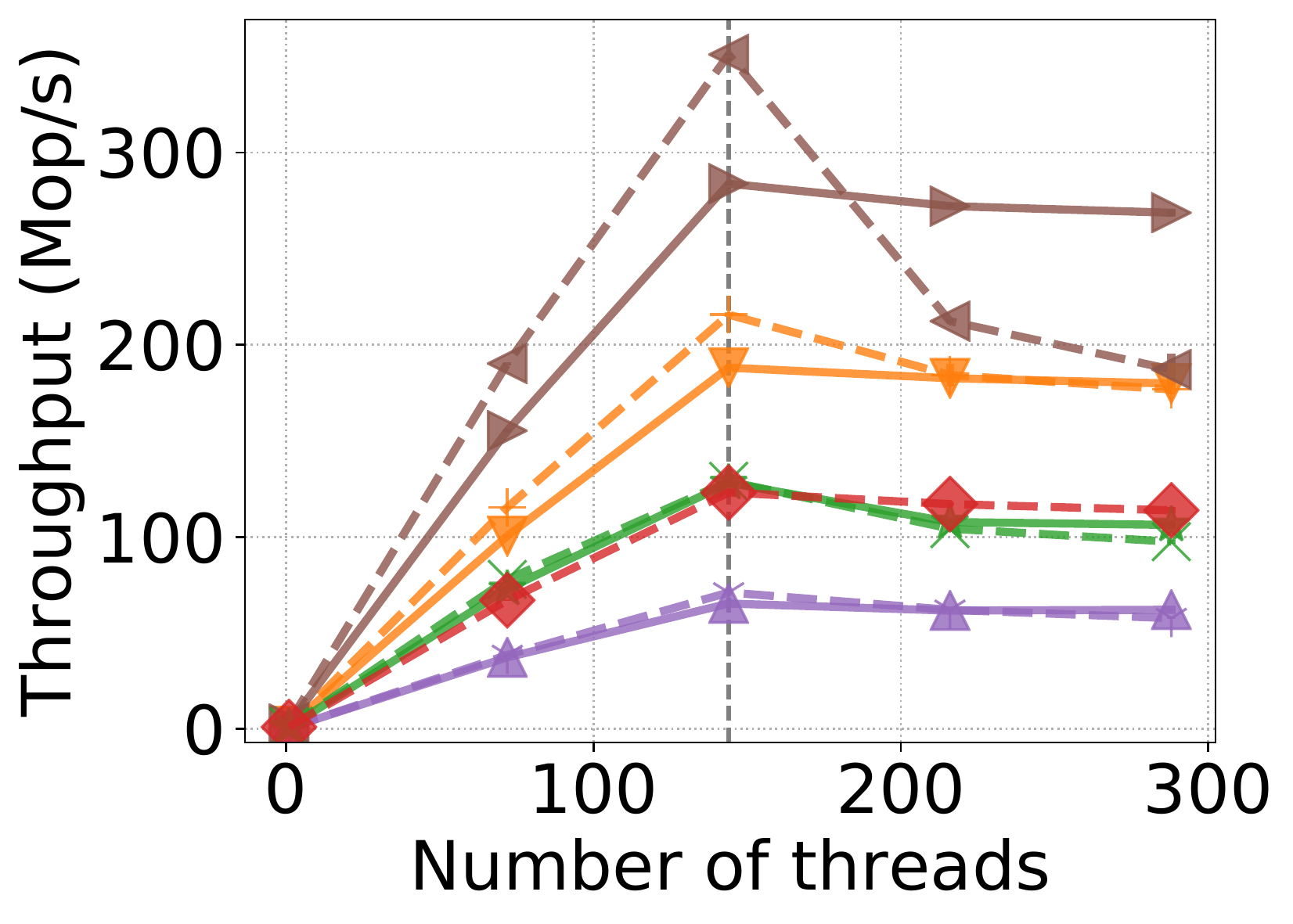}
		\captionsetup{justification=centering}
		\caption{100M keys, 50\% up., $\alpha = 0.75$}\label{subfig:rtree-threads}
	\end{subfigure}
	\begin{subfigure}{0.23\textwidth}
		\centering
		\includegraphics[width=\textwidth]{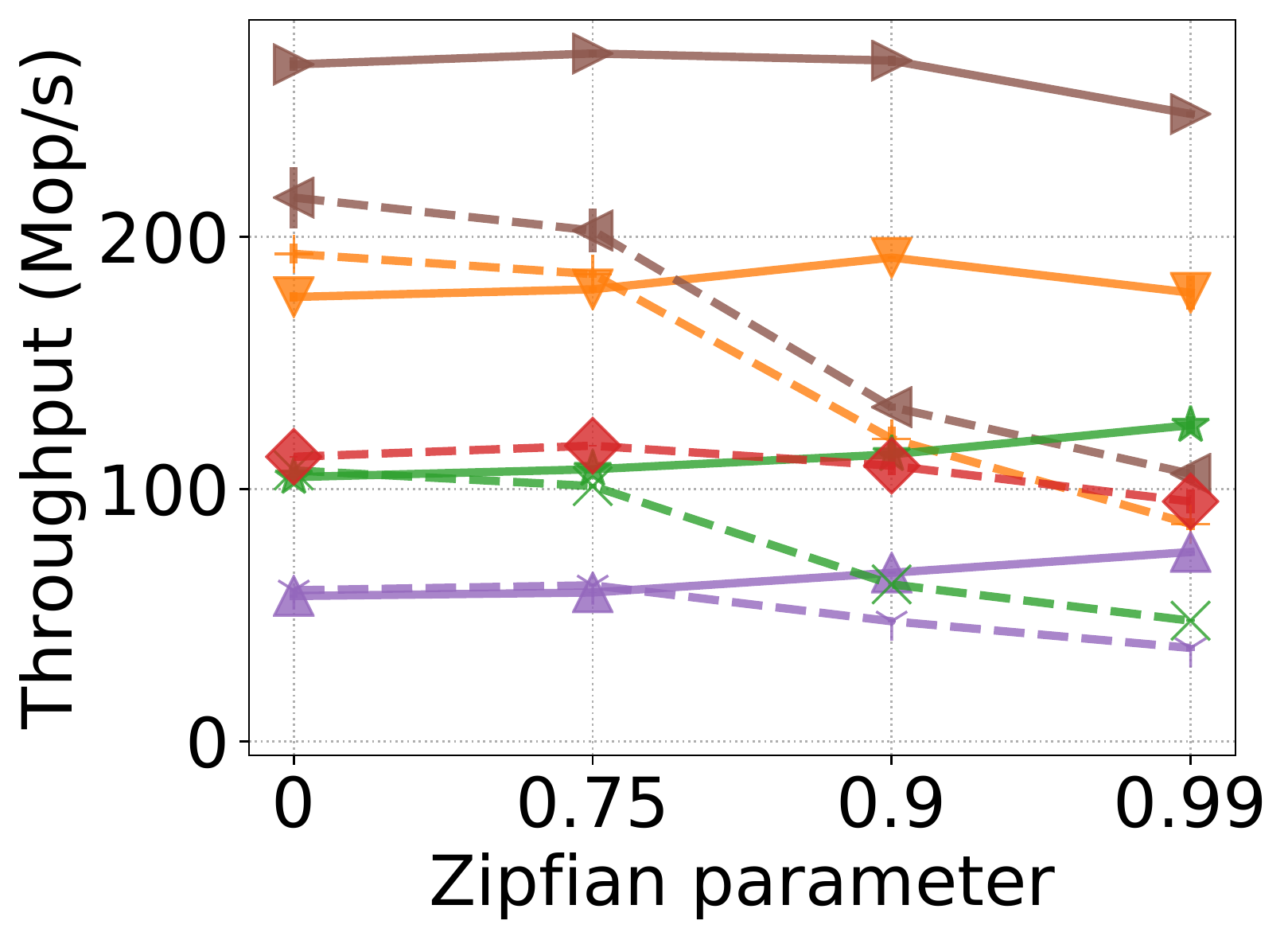}
		\captionsetup{justification=centering}
		\caption{100M keys, 216 th., 50\% up.}\label{subfig:rtree-zipf}
	\end{subfigure}
	\caption{Throughput of concurrent set data structures.}
	\label{fig:rtree-exp}
\end{figure}

\myparagraph{Binary trees} Figure~\ref{fig:tree-exp} shows the throughput of concurrent trees under a wide range of workloads. We compare our tree implementations with state-of-the-art lock-based (\bronson{}~\cite{bronson10}, \drachsler{}~\cite{drachsler14}) and lock-free (\ellen{}~\cite{EllenFRB10}, \chromatic{}~\cite{brown2014general} and \natarajan{}~\cite{Natarajan14}) binary search trees.
These implementations were obtained from the SetBench benchmarking suite~\cite{Brown18}.
\bronson{} and \chromatic{} are the only balanced tree among these implementations.
Regarding the lock-free trees, \ellen{} and \natarajan{} are implemented directly from CAS whereas \chromatic{} is implemented using the higher-level LLX/SCX primitives~\cite{BrownER13}.
Note that in all the graphs, lock-based algorithms are denoted by dotted lines and lock-free algorithms appear as solid lines.
\Guy{Why only mention ellen, chromatic and natarajan?  Perhaps drop last sentence.}
\Hao{dropped last sentence.}

Figures~\ref{subfig:tree-100M-threads}-~\ref{subfig:tree-100M-zipf-200} consider the case where the tree does not fit in cache \hedit{and Figures~\ref{subfig:tree-100K-threads}-~\ref{subfig:tree-100K-zipf} consider the case where they do. 
In out-of-cache workloads,} performance is dominated by cache misses incurred during the traversal phase. Figure~\ref{subfig:tree-100M-updates} shows that the cost of updating the tree is small compared to these cache misses, \hedit{whereas in Figure~\ref{subfig:tree-100K-updates}, increasing the percentage of updates significantly reduced throughput}.
All trees scale well, up until oversubscription (Figures~\ref{subfig:tree-100M-threads} and~\ref{subfig:tree-100K-threads}), with the exception of \drachsler{} in Figure~\ref{subfig:tree-100K-threads}.
\bronson{} is generally the fastest when tree size is large because it is better balanced compared to the other trees (many of which are only balanced in expectation due to random inserts), resulting in shorter traversals and less cache misses. 
As the zipfian parameter $\alpha$ increases, all trees except \bronson{} and \drachsler{} speed up because higher $\alpha$ means more locality and less cache misses (Figure~\ref{subfig:tree-100M-zipf-144}).
However, large $\alpha$ also means more contention. 
In the case of \bronson{} and \drachsler{}, which both use blocking strict-locks, this extra contention out-weighs the benefits of locality.
\hedit{This effect is even more severe for small trees (Figure~\ref{subfig:tree-100K-zipf}).}
\Guy{Seems odd we say a lot about 100M and nothing about 100K.}
\Hao{added more discussion of 100K graphs throughout the paragraph.}


\myparagraph{Lock-free vs blocking} Next, we compare the performance of lock-free data structures with blocking ones, with particular emphasis on \leaftreelf{} and \leaftreelb{}, the lock-free and blocking variants, respectively, of our \leaftree{}.
The overhead of \locklesslock{}s come from two main sources (1) allocating and initializing a new descriptor every time a lock is acquired, and (2) committing values to the log during critical sections.
A successful \op{insert} commits about 5 entries to the log.
This overhead is only visible in small trees with high update rates (Figures~\ref{fig:try-vs-strict} and~\ref{subfig:tree-100K-threads}).
\hedit{Across all the graphs} in Figure~\ref{fig:tree-exp}, the overhead of using \locklesslock{}s rather than traditional blocking locks is no more than 11\% \Guy{on which graph?}. 
\hedit{Furthermore,} most graphs do not show any visible overhead.

Where lock-free algorithms shine is in oversubscribed cases (e.g. 288 threads) with high contention. 
\future{\Hao{make sure we justify oversubscription somewhere. Meant to simulate a non-dedicated environments.}}
This is because a thread may get descheduled while it is partway through an update, and in a lock-free algorithm, if another thread wants to update the same location, it can simply help complete the inactive thread's update and then proceed with its own.
However, in a blocking data structure, the new thread will have to either wait for the inactive thread to be scheduled again and release its lock, or yield and context switch, both of which are expensive.
This effect can be seen in the right side of Figures~\ref{subfig:tree-100M-zipf-200} and~\ref{subfig:tree-100K-zipf} and the left side of Figure~\ref{subfig:tree-sizes} where the four lock-free trees outperform the three blocking trees.
In particular, \leaftreelf{} outperforms \leaftreaplb{} by up to 2.4x in Figure~\ref{subfig:tree-sizes}.

\myparagraph{Other set datatypes} In Figure~\ref{fig:rtree-exp}, \arttree{}, \leaftreap{}, \abtree{} and \hashtable{}, generally follow the same pattern as \leaftree{}.
That is, lock-free versions outperform their blocking counterparts in oversubscribed, high contention scenarios (right side of Figure~\ref{subfig:rtree-zipf}), by up to 2.5x in the case of the \hashtable{} and 2x for the \arttree{}.
In non-oversubscribed scenarios (left size of Figure~\ref{subfig:rtree-threads}), the overhead of using \locklesslock{}s is small, especially for \abtree{} and \leaftreap{}. 
\hedit{The overhead of \locklesslock{}s is highest in the \hashtable{} because its search time (i.e. fraction of time spent outside of the critical section) is small and hence the overhead for the locked part plays a larger role.}
Figure~\ref{fig:rtree-exp} also compares our data structures with Srivastava's CoPub-ABtree~\cite{srivastava2021extremely}, a state-of-the-art blocking (a,b)-tree.
Our lock-free \abtree{} performs similarly to \op{srivastava\_abtree} in most cases but is up to 32\% faster at the right of Figure~\ref{subfig:rtree-zipf}. 
\Guy{Might want to mention that blocking version of hash table has largest difference because its search time (not affected by helping) is small and hence the overhead for the locked part plays a larger role.} \Hao{Added a sentence about this}


\begin{figure}
	\centering
	\begin{subfigure}{0.37\textwidth}
		\centering
		\includegraphics[width=\textwidth,trim={2.3cm 0 0 0},clip]{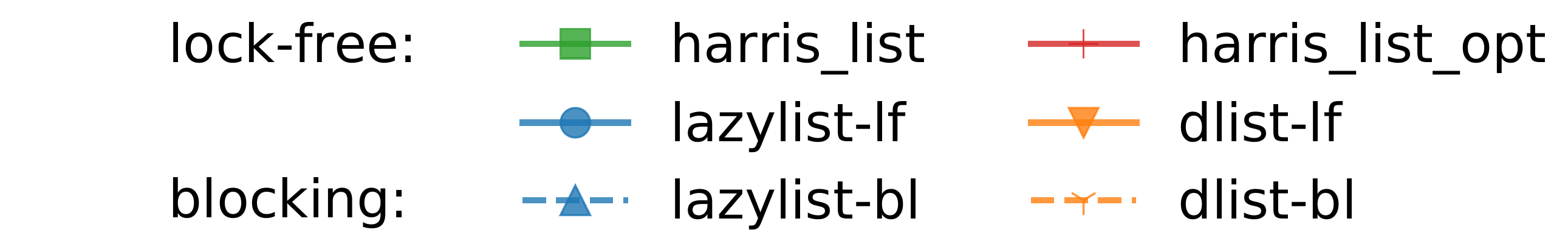}
	\end{subfigure}
	\begin{subfigure}{0.23\textwidth}
		\centering
		\includegraphics[width=\textwidth]{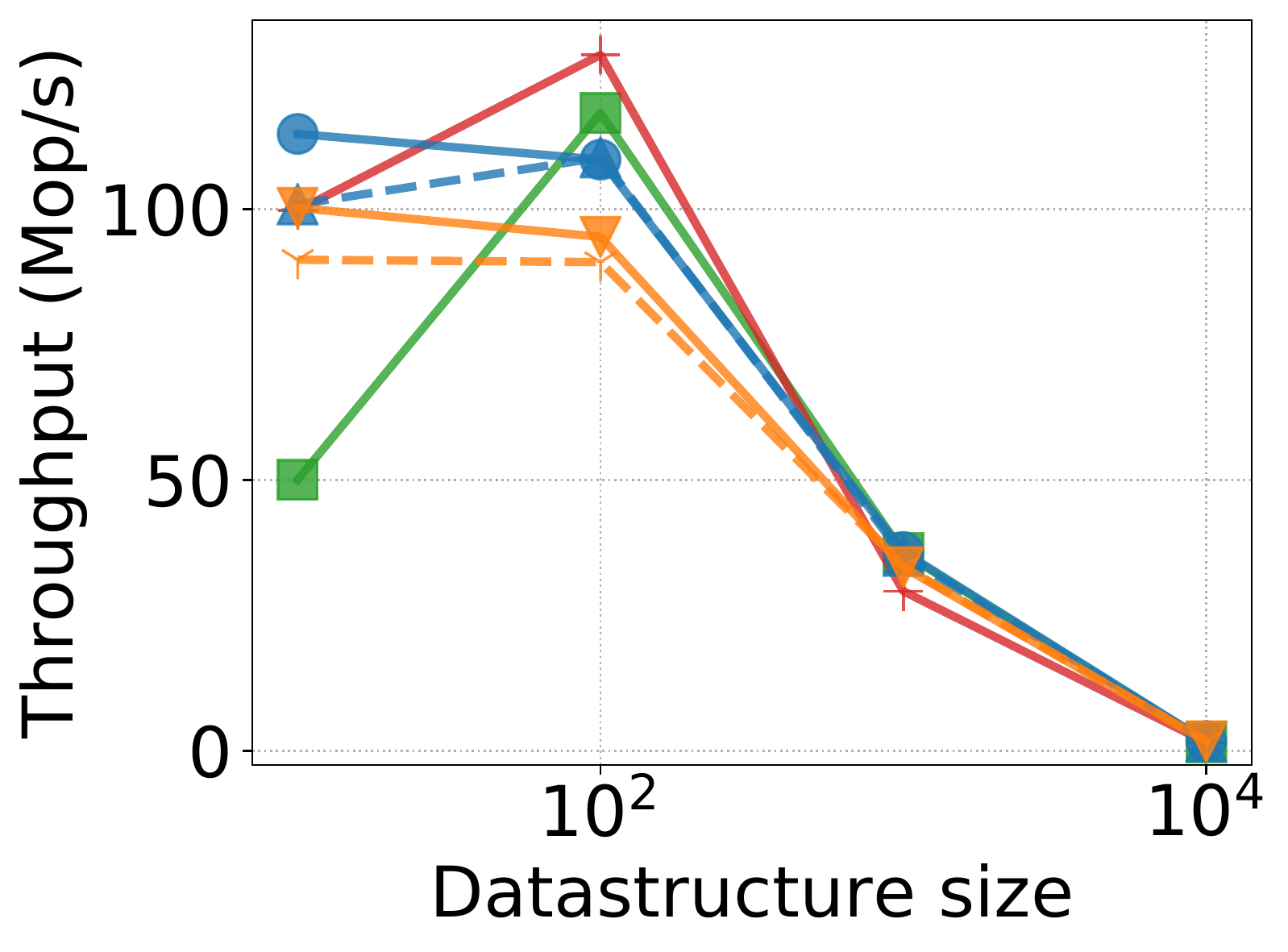}
		\captionsetup{justification=centering}
		\caption{144 th., 5\% up., $\alpha = 0.75$}\label{subfig:list-sizes}
	\end{subfigure}
	\begin{subfigure}{0.23\textwidth}
		\centering
		\includegraphics[width=\textwidth]{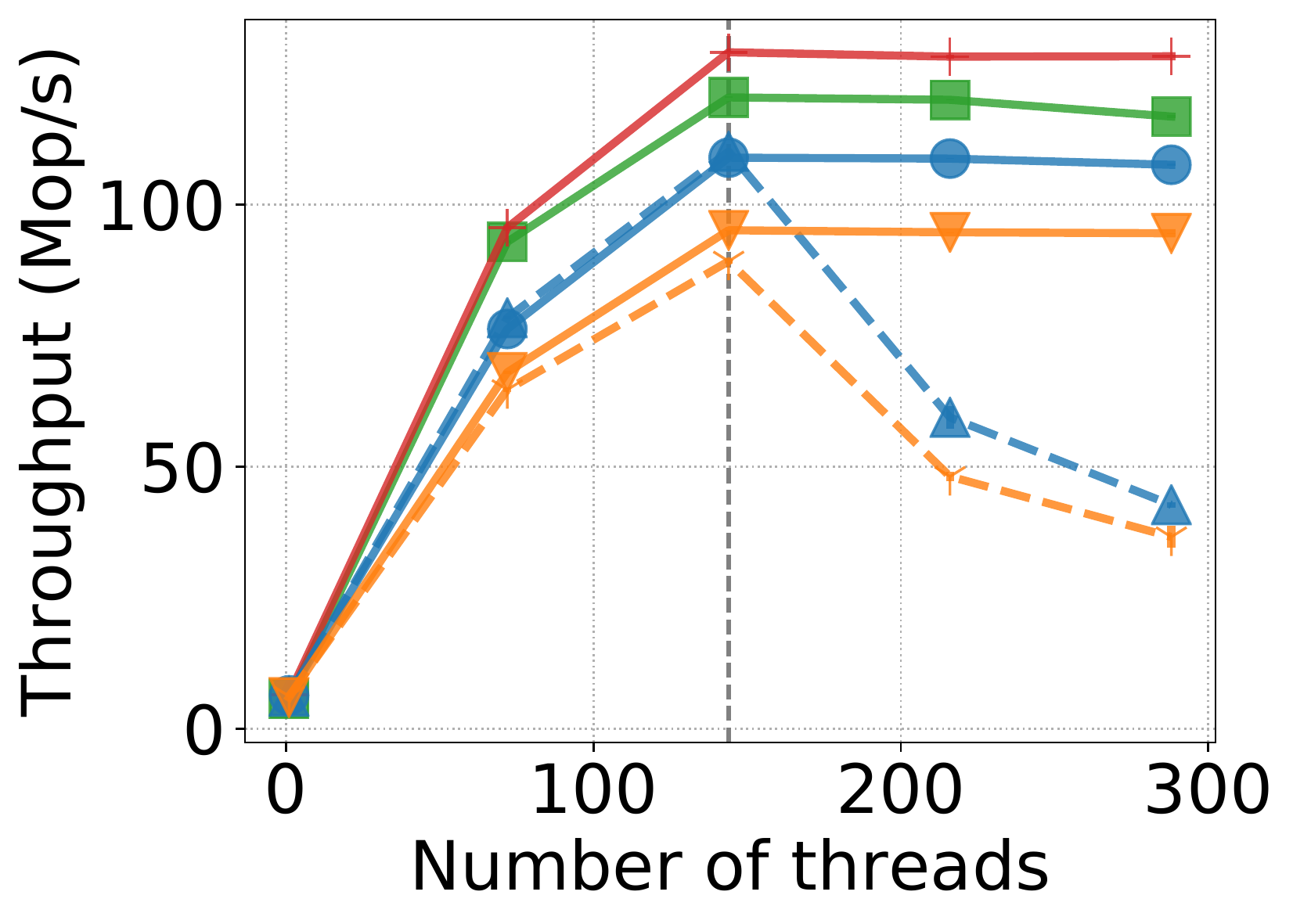}
		\captionsetup{justification=centering}
		\caption{100 keys, 5\% up., $\alpha = 0.75$}\label{subfig:list-threads}
	\end{subfigure}
	\caption{Throughput of singly and doubly linked lists. The `bl' and `lf' suffixes represent the blocking and lock-free version of our locks, respectively.}
	\label{fig:list-exp}
\end{figure}

\myparagraph{Linked List Experiments} Figure~\ref{fig:list-exp} compares doubly and singly linked lists written using our lock-free locks (\dlist{} and \lazylist{}, respectively) with Harris's lock-free singly linked list~\cite{harris2001pragmatic} (\op{harris\_list}), and an optimized version of Harris's list where \op{find} operations do not perform any helping~\cite{DGT15} (\op{harris\_list\_opt}).
In most cases, our lock-free \lazylist{} is slower than \op{harris\_list\_opt} by about 16\% because the descriptors in \op{harris\_list\_opt} are optimized to simply be flags. Interestingly, the lock-free versions of \dlist{} and \lazylist{} outperform their corresponding blocking versions even without oversubscription on small lists (left of Figure~\ref{subfig:list-sizes}).

The pseudo-code for \dlist{} was presented earlier in Algorithm~\ref{alg:list} and Figure~\ref{fig:list-exp} show that this simple algorithm performs well.
The overhead of maintaining back pointers is only about 13\% (comparing \dlist{} with \lazylist{}).

	\section{Conclusion}

We presented a mechanism for implementing lock-free locks, and a library-based implementation.  It is the first such library implementation of lock-free locks we know of.  The approach is practical in two senses.  Firstly, in terms of performance it is competitive with state-of-the-art lock-free and lock-based data structures.  Secondly, using the library requires very few changes to existing lock-based implementations---basically wrapping shared values in a mutable, and using the \flock{} lock structure and memory management.  In terms of functionality it significantly extends previous suggested approaches to lock-free locks, supporting memory management, races, and \trylock{}s.

We separate out the idea of idempotent blocks of code (thunks) and present a general and efficient approach along with
a C++-based library to support them.  The approach supports arbitrary code with load, stores and CAMs on shared locations, as well as memory allocation and retirement from a shared pool.  A thunk using the approach can be run any number of times with instructions interleaved in any way while behaving like it ran once.   The approach uses a shared log for each thunk so that separate runs of the thunk see the same result.  The idempotent construction could be of independent interest.

We implemented several data structures using the approach.  With regards to the opening question of whether to be lock-free or not to be, the experiments clearly indicate the advantage of lock-freedom when processors are oversubscribed.  Our experiments are some of the first on concurrent data structures that study \hedit{this effect.}
The experiments also show that the overhead of being lock-free for our structures is relatively small (rarely more than 10\%) and often hardly noticeable.

\begin{acks}
	We thank the anonymous referees for their comments. This work was supported by the National Science Foundation grants CCF-1901381, CCF-1910030, and CCF-1919223.
\end{acks}
	
	\bibliographystyle{ACM-Reference-Format}
	\bibliography{strings,biblio}
        \appendix

\section{Proof of Theorem~\ref{theorem:idempotent}}

\begin{proof}
	Given an execution $E$ consisting of runs of $T$ interleaved with arbitrary other steps on shared data, we will construct a subsequence $E'$ of $E|T$ that satisfy the criteria from Definition~\ref{def:idempotence} (with the relaxation that retire operations in $E'$ are allowed to appear later than they would have in a single run of $T$).
	Throughout the proof, we will refer to load, store, CAM, allocate, retire as operations, and executions of primitive shared memory instructions such as read, write, and CAS as steps.
	
	We begin by viewing the execution at the level of operations. We show by induction that all runs of $T$ execute the same sequence of operations with the same arguments and return values.
	As the base case, note that all processes that execute $T$ start with the same local variables, and $T$ takes no arguments. Therefore, they begin the execution in the same state. 
	As the inductive hypothesis, assume that the first $k-1$ operations executed by $T$ are the same across all runs and have the same arguments and return values. 
	Consider the $k$th operation $O_k$.
	Since all previous operations returned the same value across all runs, then $O_k$ is the same operation and is called with the same arguments across all runs. 
	Note furthermore that if $O_k$ executes line~\ref{line:commit}, the CAS on that line is successful in exactly one instance. 
	All processes executing $O_k$ use position $k$ to access the log, and no process executing a different operation uses position $k$. Therefore, before the first execution of line~\ref{line:commit} for $O_k$, \texttt{log[k] = empty}. 
	Since we assume \texttt{empty} is never written in any allocated variable, 
	the new value of the CAS on line~\ref{line:commit} will never be \texttt{empty}. Therefore, the first instance of that CAS will be successful, and all others will fail. 
	Therefore, if $O_k$ is a \texttt{load} or an \texttt{allocate}, since those operations return the value read from \texttt{log[k]} after the first CAS on line~\ref{line:commit} for $O_k$, all its instances will return the same value. Note that all other operations do not return a value, so the claim holds.
	
	Next, we construct the subsequence $E'$ by picking steps so that each operation $O_k$ appears to only run once.
	We will ensure that 
	all steps picked from runs of $O_{k}$ appear before those picked from runs of $O_{k+1}$,
	except when $O_k$ is a \texttt{retire} operation in which case its call to \texttt{sysRetire} may appear later.
	For each operation $O_k$, consider the run that executes the CAS on line~\ref{line:commit} first.
	We pick a prefix of that run, starting from the beginning of $O_k$ up to when it executes line~\ref{line:commit} (inclusive), to be part of $E'$.
	As shown in the previous paragraph, executions of the CAS on line~\ref{line:commit} by other runs of $O_k$ will return false.
	Next, we pick the first execution of line~\ref{line:commit-load} by any run of $O_k$ to be part of $E'$, and we pick the remaining steps differently depending on what type of operation $O_k$ is.
	\texttt{load} operations do not perform any more shared memory steps so we are done.
	Let $r$ be a run of $O_k$ that is consistent with the sequence of steps we have picked so far for $E'$.
	Since we picked the successful instance of line~\ref{line:commit}, \texttt{isFirst} is set to true for $r$.
	Therefore, if $O_k$ is an \texttt{allocate}, then $r$ will not execute any more shared memory steps after line~\ref{line:commit-load}, so $E'$ contains all of $r$'s steps.
	If $O_k$ is an \texttt{retire}, then whichever run executed the successful CAS on line~\ref{line:commit} will eventually execute a \texttt{sysRetire} on line~\ref{line:retire}, and we pick that \texttt{sysRetire} to be part of $E'$ (if it exists in $E$).
	Note that this \texttt{sysRetire} may appear in $E'$ after steps by future operations and this is allowed by our relaxed idempotence definition.	
	If $O_k$ is a \texttt{store}, then we pick the first execution of the \texttt{CAS} on line~\ref{line:storeCAS} to be part of $E'$.
	All executions of this CAS by future runs of $O_k$ will return false because \texttt{oldV} was previously stored in \texttt{val} and we assume mutable types are ABA-free.
	Finally, suppose $O_k$ is a \texttt{CAM}.
	Since the value of \texttt{check} on line~\ref{line:CAMcheck} was read from the log on line~\ref{line:commit-load}, all runs of $O_k$ will have the same value for \texttt{check}.
	Therefore, either all runs will execute the \texttt{CAS} on line~\ref{line:CAMCAS} or none of them will.
	If they execute the \texttt{CAS}, then we pick the first such step to be part of $E'$, just like for \texttt{store}s.
	Otherwise, $O_k$ performs no more shared memory steps and we are done.
	
	Picking steps in this manner ensures that $T$ appears to run once in $E'$ and if there is a finished run of $T$ in $E$, then the last step of the first finished run will be the end of $E'$ (with the exception of \texttt{sysRetire}).
	Furthermore, the steps in $E|T$ that we did not pick have no effect on shared memory so removing them still leaves a valid history.
	This is the case for any removed CAS operation because they all return false.
	Also, memory locations allocated by removed \texttt{sysAllocate} operations are never used since they are never committed to the log. 
	Finally, the \texttt{sysFree}s that were removed correspond to the removed \texttt{sysAllocate} operations.
	Therefore, removing all of $T$'s steps from $E$ other than those in $E'$ leaves a valid history.
\end{proof}

\end{document}